\documentclass[generic,preprint,10pt]{imsart}

\setattribute{journal}{name}{}
\renewcommand\baselinestretch{1.0}

\RequirePackage[OT1]{fontenc}
\RequirePackage{amsthm,amsmath,natbib}
\RequirePackage[colorlinks,citecolor=blue,urlcolor=blue,breaklinks=true]{hyperref}
\RequirePackage{hypernat}
\RequirePackage{comment}
\RequirePackage{pzccal}
\RequirePackage{graphicx,subfigure,latexsym,amssymb}
\RequirePackage{float,epsfig,multirow,rotating}
\RequirePackage{upgreek,wrapfig}

\arxiv{math.PR/0000000}
\startlocaldefs
\newcommand {\ctn}{\citet} 
\newcommand {\ctp}{\citep}       

\numberwithin{equation}{section}
\theoremstyle{plain}
\newtheorem{theorem}{Theorem}[section]
\newtheorem{definition}{Definition}[section]
\newtheorem{remark}{Remark}[section]
\newtheorem{proposition}{Proposition}[section]
\newtheorem{lemma}{Lemma}[section]

\newcommand{\bell}{\boldsymbol{\ell}}

\newcommand{\bg}{\boldsymbol{g}}
\newcommand{\bof}{\boldsymbol{f}}

\newcommand{\bxi}{\boldsymbol{\xi}}

\newcommand{\bSigma}{\boldsymbol{\Sigma}}
\newcommand{\bsigma}{\boldsymbol{\sigma}}
\newcommand{\bPsi}{\boldsymbol{\Psi}}

\newcommand{\bD}{\boldsymbol{D}}

\newcommand{\bV}{\boldsymbol{V}}
\newcommand{\bm}{\boldsymbol{m}}
\newcommand{\bG}{\boldsymbol{G}}

\newcommand{\bM}{\boldsymbol{M}}

\newcommand{\bA}{\boldsymbol{A}}
\newcommand{\bS}{\boldsymbol{S}}

\newcommand{\bq}{\boldsymbol{q}}
\newcommand{\bQ}{\boldsymbol{Q}}

\newcommand{\bW}{\boldsymbol{W}}

\newcommand{\bv}{\boldsymbol{v}}
\newcommand{\bs}{\boldsymbol{s}}

\newcommand{\bx}{\boldsymbol{x}}
\newcommand{\bX}{\boldsymbol{X}}

\newcommand{\bY}{\boldsymbol{Y}}
\newcommand{\bZ}{\boldsymbol{Z}}

\newcommand{\bL}{\boldsymbol{L}}
\newcommand{\bzero}{\boldsymbol{0}}

\endlocaldefs

\begin{document}
\renewcommand\baselinestretch{1.0}

\begin{frontmatter}
\title{Deep Bayesian Supervised Learning given Hypercuboidally-shaped,
Discontinuous Data, using Compound Tensor-Variate $\&$ Scalar-Variate Gaussian
Processes
}
\runtitle{Compound Tensor $\&$ Scalar-Variate GPs}

\begin{aug}
\author{
{\fnms{Kangrui} \snm{Wang}\thanksref{t2,m2}\ead[label=e2]{kw202@le.ac.uk}},
{\fnms{Dalia} \snm{Chakrabarty}\thanksref{t1,m1}\ead[label=e1]{d.chakrabarty@lboro.ac.uk}}
},
\thankstext{t2}{PhD student in Department of Mathematics, University of Leicester} 
\thankstext{t1}{Lecturer in Statistics, Department of Mathematical Sciences,
  Loughborough University}

\runauthor{Wang $\&$ Chakrabarty}

\affiliation{University of Leicester, Loughborough University}

\address{\thanksmark{m2} Department of Mathematics\\
University of Leicester\\
Leicester LE1 3RH,
U.K.\\
\printead*{e2}
}

\address{\thanksmark{m1} Department of Mathematical Sciences\\
Loughborough University\\
Loughborough LE11 3TU,
U.K.\\
\printead*{e1}
}

\end{aug}

\begin{abstract} {We undertake Bayesian learning of the high-dimensional 
functional relationship
between a system parameter vector and an observable, that is in general
tensor-valued. The ultimate aim is Bayesian inverse prediction of the system
parameters, at which test data is recorded. We attempt such learning given
hypercuboidally-shaped data that displays strong discontinuities, rendering
learning challenging. We model the sought high-dimensional function, with a
tensor-variate Gaussian Process (GP), and use three independent ways for
learning covariance matrices of the resulting likelihood, which is
Tensor-Normal. We demonstrate that the discontinuous data
demands that implemented covariance kernels be non-stationary--achieved by
modelling each kernel hyperparameter, as a function of the sample function of
the invoked tensor-variate GP. Each such function can be shown to be
temporally-evolving, and treated as a realisation from a distinct
scalar-variate GP, with covariance described adaptively by collating
information from a historical set of samples of chosen sample-size. We prove
that deep-learning using 2-''layers'', suffice, where the outer-layer
comprises the tensor-variate GP, compounded with multiple scalar-variate GPs
in the "inner-layer", and undertake inference with Metropolis-within-Gibbs.
We apply our method to a cuboidally-shaped, discontinuous, real dataset, and
subsequently perform forward prediction to generate data from our model, given
our results--to perform model-checking.
 }
\end{abstract}



\end{frontmatter}

\renewcommand\baselinestretch{1.0}
{
\section{Introduction}
\noindent
Statistical modelling allows for the learning of the relationship between two
variables, where the said relationship is responsible for generating the data
available on the variables. Thus, let $\bX$ be a random variable that represents a behavioural or structural
parameter of the system, and $\bY$ is another variable that bears influence on
$\bX$ s.t. $\bY=\bof(\bX)$, where
the functional relation $\bof(\cdot)$ that we seek to learn, is itself a
random structure, endowed with
information about the error made in predicting the values of $\bY$ (or $\bX$)
at which the noise-included measurement of $\bX$ (or $\bY$) has been realised.
Such a function can be modelled as a realisation from an adequately chosen
stochastic process.
In general, either or both variables could be tensor-valued, such that, data
comprising measurements of either variable, is then shaped as a
hypercuboid. Typically, the structure/behaviour of a system is parametrised
using a set of scalar-valued parameters, (say $d$ number of such parameters),
which can, in principle be collated into a $d$-dimensional vector. Then $\bX$
is typically, the system parameter vector.  The other, observed variable
$\bY$, can be tensor-valued in general. There are hypercuboidally-shaped data
that show up in real-world applications, \ctp{mardia_book, bijma_face,
  werner_cvs, theobald_covs, fuhrman}. For example, in computer vision, the
image of one person might be a matrix of dimensions $a\times b$, i.e. image
with resolution of $a$ pixels by $b$ pixels. Then, repetition across $n$
persons inflates the data to a cuboidally-shaped dataset. Examples of handling
high-dimensional datasets within computer vision exist \ctp{ian_face,
  fu_thesis, pang, wang_book, qiang}. In health care, the $p$ number of health
parameters of $n$ patients, when charted across $k$ time-points, again
generates a high-dimensional data, which gets further enhanced, if the
experiment involves tracking for changes across $\ell$ groups of $n$ patients
each, where each such group is identified by the level of intervention
\ctp{chari, clarke, oberg, chari2, sarakar_phd, wang_healthcare, fan2107}.
Again, in ecological datasets, there could be $n$ spatial locations at each of
which, $p$ traits of $k$ species could be tracked, giving rise to a
high-dimensional data \ctp{leitao, warton, dunstanscott}.

It is a shortcoming of the traditional modelling strategies that we treated
these groupings in the data as independent--or for that matter, even the
variation in parameter values of any group across the $k$ time points, is
ignored, and a mere snapshot of each group is traditionally considered, one at
a time. In this work, we advance a method for the consideration of parameters
across all relevant levels of measurement, within one integrated framework, to
enable the learning of correlations across all such levels, thus permitting
the prediction of the system parameter vector, with meaningful uncertainties,
and avoid information loss associated with categorisation of data.

While discussing the generic methodology that helps address the problem of
learning the inter-variable relationship $\bof(\cdot)$, given general
hypercuboid-shaped data, we focus on developing such learning when this data
displays discontinuities. In such a learning exercise, the inter-variable functional relation
$\bof(\cdot)$, needs to be modelled using a
high-dimensional stochastic process (a tensor-variate Gaussian Process, for
example), the covariance function of which is non-stationary. The correlation
between a pair of data slices, (defined by two such measured values of $\bY$,
each realised at two distinct values of the system parameter $\bX$), is
sometimes parametrically modelled as a function of the distance between the
values of the system parameter at which these slices are realised,
i.e. ``similarity'' in values of $\bY$ can be modelled as a function of
``similarity'' in the corresponding $\bX$ values. However, if there are
discontinuities in the data, then such a mapping between ``similarities'' in
$\bX$ and $\bY$ no longer holds. Instead, discontinuities in data call
for a model of the correlation that adapts to the discontinuities in the
data. We present such correlation modelling in this paper, by modelling each
scalar-valued hyperparameter of the correlation structure of the
high-dimensional stochastic process, as a random function of the sample path
of that process; this random function then, can itself be modelled as a
realisation of a scalar-variate stochastic process--a scalar-variate Gaussian
Process (GP) for example (Section~\ref{sec:model}).  

Thus, the learning of $\bof(\cdot)$ is
double-layered, in which multiple scalar-variate GPs inform a high-dimensional
(tensor-variate) GP. Importantly, we show below (Section~\ref{sec:suffice}) that no more
than 2 such layers in the learning strategy suffice. Thus, the data on the
observable $\bY$ can be shown to be sampled from a compound tensor-variate and
multiple scalar-variate Gaussian Processes.

Acknowledgement of non-stationarity in correlation learning is not new
\ctp{paciorek}.  In some approaches, a transformation of the space of the
input variable is suggested, to accommodate non-stationarity \ctp{sampson,
 snoek, ohagan}. When faced with learning the dynamically varying covariance
structure of time-dependent data, others have resorted to
learning such a covariance, using Generalised Wishart Process
\ctp{wilson}. In another approach, latent parameters that bear information on
non-stationarity, have been modelled with GPs and learnt simultaneously with
the sought function \ctp{tolvanen}, while others have used multiple GPs to
capture the non-stationarity \ctp{gramacy, heinonen}.  However, what has not
been presented, is a template for including non-stationarity in
high-dimensional data, by nesting lower-dimensional Gaussian Processes with
distinct covariances, within a tensor-variate GP (Section~\ref{sec:model} and
Section~\ref{sec:kernel}), using a Metropolis-within-Gibbs inference scheme
(Section~\ref{sec:inference}), to perform with-uncertainties learning of a
high-dimensional function, given discontinuities that show up in the
hypercuboidally-shaped datasets in general, and illustration of the method on
a cuboidally-shaped, real-world dataset (Section~\ref{sec:application},
Section~\ref{sec:prediction}). This is what we introduce in this paper. Our
model is capacitated to learn the temporally-evolving covariance of
time-dependent data (Section~\ref{sec:temporal}), if such is the data at hand, but the focus of our
interest is to follow the learning of the sought tensor-valued functional
relation between a system parameter vector and a tensor-valued observable,
with inverse Bayesian prediction of the system parameter values, at which test
data on the observable is measured (Section~\ref{sec:results},
Section~\ref{sec:prediction}). Additionally, flexibility of our model design
permits both inverse and forward predictions. So we also
predict new data at chosen system parameter values given our model and
results, and perform model checking, by comparing such generated data against
the empirically observed data (Section~3 of Supplementary Materials).

\section{Model}
\label{sec:model}
\noindent
Let system parameter vector $\bS\in{\boldsymbol{\cal X}}\subseteq{\mathbb
    R}^d$, be affected by observable $\bV$, where $\bV$ is ($k-1$-th ordered)
  tensor-valued in general, i.e.  is $\bV\in{\boldsymbol{\cal
      Y}}\subseteq{\mathbb R}^{m_1\times m_2\times\ldots\times m_{k-1}}$,
  $m_i\in{\mathbb Z},\forall\:i=1,\ldots,k-1$.
That $\bV$ bears influence on $\bS$ suggests
the relationship $\bV=\bxi(\bS)$
where $\bxi:{\boldsymbol{\cal X}}\subseteq{\mathbb R}^d\longrightarrow {\boldsymbol{\cal Y}}\subseteq{\mathbb R}^{m_1\times
  m_2\times\ldots\times m_{k-1}}$. 


\begin{definition}
  \label{defn:defn1} 
{We define functional relationship $\bxi(\cdot)$, between $\bS$ and
    $\bV$, as a ``tensor-valued function'', with
    $\displaystyle{\prod\limits_{i=1}^{k-1} m_i}$-number of component
    functions, where these components suffer inter-correlations. Thus, the
    learning of $\bxi(\cdot)$ is equivalent to learning the component
    functions, inclusive of learning the correlation amongst these component
    functions.

Inverse of 
$\bxi(\cdot)$, is defined as the tensor-valued function of same
dimensionalities as $\bxi(\cdot)$, comprising inverses of each
component function of $\bxi(\cdot)$, assuming inverse of each component
function exists.}
\end{definition}


The inversion of the sought function $\bxi(\cdot)$--where
$\bV=\bxi(\bS)$--allows for the forward prediction of $\bv^{(new)}$ given a
measured value $\bs^{(new)}$ of $\bS$, as well as for the inverse prediction
of the value of $\bS$ at which a given measurement of $\bV$ is recorded.
It may be queried: why do we undertake the seemingly more difficult learning
of the tensor-valued $\bxi(\cdot)$ (that outputs the tensor $\bV$), than of
the vector-valued $\bg(\cdot)$ (that outputs the vector $\bS$). We do this,
because we want to retain the capacity of predicting both new data at a given
value of the system parameter ($\bS$), as well as predict the system parameter
at which a new measurement of the observable $\bV$ is realised.
\begin{remark}
{
If we had set ourselves the task of learning $\bg(\cdot)$, where
$\bg(\bV)=\bS$, i.e. $\bg(\cdot)$ is a ``vector-valued'' function,
and therefore lower dimensional with fewer number of
component functions than the tensor-valued $\bxi(\cdot)$--we could not have
predicted value of $\bV$ at a given $\bs$. 
The $d$-dimensional vector-valued inverse function $\bg^{-1}(\cdot)$ cannot yield a value of the $\displaystyle{\prod\limits_{i=1}^{k-1}
  m_i}$ number of components of the tensor $\bV$ at this given $\bS$, if
$\displaystyle{\prod\limits_{i=1}^{k-1} m_i} > d$. 
}
\end{remark}
 
The learning of the function $\bxi(\cdot)$, uses the training data
${\bf D}:=\{(\bs_i,\bv_i)\}_{i=1}^N$. Conventional prediction of
$\bS=\bs^{(test)}$, at which test data $\bv^{(test)}$ on $\bV$ is realised,
suggests: $\bs^{(test)} := \bxi^{-1}(\bV)\vert_{\bv^{(test)}}$. 
\begin{itemize}
\item However, this
there is no objective way to include the uncertainties learnt
in the learning of the function $\bxi(\cdot)$, to propagate into the
uncertainty of this prediction. This underpins an advantage of Bayesian
prediction of one variable, given test data on the other, subsequent to
learning of $\bxi(\cdot)$ using training data ${\bf D}$. 
\item Conventional fitting
methods (such as fitting with splines, etc), also fumble 
when measurements of both/either of the
r.v.s $\bS$ and $\bV$, are accompanied by measurement errors; in light of
this, it becomes difficult to infer the function that fits the data the
best. In fact, the uncertainty in the learning of the sought function is also
then difficult to quantify. 
\item Secondly, there is no organic way of quantifying the smoothness of the
  sought $\bxi(\cdot)$, in th econventional approach. Ideally, we would prefer
  to learn this smoothness from the data itself. However, there is nothing
  intrinsic to the fitting-with-splines/wavelets method that can in principle,
  quantity the smoothness of the curve, given a training data.
\item Lastly, when
$\bV$ is an r.v. that is no longer a scalar, but higher-dimensional (say
tensor-valued in general), fitting with splines/wavelets starts to become
useless, since in such cases of sought tensor-valued function $\bxi(\cdot)$
(in general), the component functions of $\bxi(\cdot)$ are correlated, but
methods such as parametric fitting approaches, cannot capture such
correlation, given the training data. As we have remarked above, such
correlation amongst the components functions of $\bxi(\cdot)$ is the same
correlation structure amongst the components of the tensor-valued $\bV$--so in
principle, the sought correlation can be learnt from the training data.
\end{itemize}

In light of this, we identify a relevant Stochastic Process that can give a
general, non-restrictive description of the sought function $\bxi(\cdot)$--a
Gaussian Process for example. The joint probability density of a set of
realisations of a sampled $\bxi(\cdot)$, is then driven by the Process under
consideration, where each such realisation of the function, equals a value of
the output variable $\bV$. Thus, the joint also represents the likelihood of
the Process parameters given the relevant set of values of $\bV$, i.e. the
data. We impose judiciously chosen priors, to write the posterior probability
density of the Process parameters given the data. Generating samples from this
posterior then allows for the identification of the 95$\%$ HPD credible
regions on these Process parameters, i.e. on the learnt function
$\bxi(\cdot)$. It is possible to learn the smoothness of the function
generated from this Process, via kernel-based parameterisation of the
covariance structure of the GP under consideration. Thus, we focus on the
pursuit of adequate covariance kernel parametrisation.

\begin{proposition} {When possible, covariance matrices of the GP that is
    invoked to model the sought function $\bxi(\cdot)$, are kernel-parametrised
    using stationary-looking kernel functions, hyperparameters of
    which are modelled as dependent on the sample paths (or rather sample
    functions) of this GP. We show below (Lemma~\ref{lemma:1}) that such a
    model can address the anticipated discontinuities in data.}
\end{proposition}


As LHS of equation $\bV=\bxi(\bS)$ is $k-1$-th ordered
tensor-valued, $\bxi(\cdot)$ is tensor-variate function of equal
dimensionalities. So we model
$\bxi(\cdot)$  as a realisation from a tensor-variate GP. 
\begin{definition}
{Modelling $\bxi(\cdot)$ as sampled from a tensor-variate GP, where the
  $k-1$-th ordered tensor-valued variable $\bV=\bxi(\bS)$, we get that the 
joint probability of the set of values of sampled function $\bxi(\cdot)$, at each of the $n$ design
points $\bs_1,\ldots \bs_n$ (that reside within the training data ${\bf
  D}=\{(\bs_i,\bv_i)\}_{i=1}^n$), follows
the $k$-variate Tensor Normal distribution \ctp{kolda, richter, tensor,
  manceur}:
$$[\bxi(\bs_1),\ldots,\bxi(\bs_n)]\sim{\cal TN}(\bM,\bSigma_1,\ldots,\bSigma_k),$$
where mean of this density is a $k$-th ordered mean tensor $\bM$ of
dimensions $m_1\times\ldots\times m_k$, and $\bSigma_j$ is the $m_j\times
m_j$-dimensional, $j$-th covariance matrix; $j=1,\ldots,k$. In other words, 
likelihood of $\bM,\bSigma_1,\ldots, \bSigma_k$ given ${\bf D}$ is the
$k$-variate Tensor Normal density:
\begin{equation}
{\cal L}(\bM,\bSigma_1,...,\bSigma_k\vert {\bf D}) \propto \exp(-\Vert ({\bf
  D}_{\bV} -\bM)\times_1 \bA_1^{-1} \times_2 \bA_2^{-1} ... \times_k \bA_k^{-1} \Vert^2/2),
\label{eqn:eqn1}
\end{equation}   
where $n$ observed values of the $k-1$-th dimensional tensor-valued $\bV$
are collated to form the $k$-th ordered tensor ${\bf D}_{\bV}$. The notation
$\times_j$ in Equation~\ref{eqn:eqn1} presents the $j$-mode product of a
matrix and a tensor \ctp{oseledets2011tensor}. Here $\bA_j$ is the unique 
square-root of the positive definite covariance matrix $\bSigma_j$, i.e. 
$\bSigma_j= \bA_j \bA^{ T}_j$.}
\end{definition}
One example of a computational algorithm that can be invoked to
realise such a square root of a matrix, is Cholesky decomposition\footnotemark. 

\footnotetext{
The covariance tensor of this $k$-th order Tensor Normal distribution, has been
decomposed into $k$ different covariance matrices by Tucker decomposition,  
\ctp{hoff_2011, manceur, kolda, Xu}, to yield the $k$ number of covariance matrices,
$\bSigma_1,\ldots, \bSigma_k$, where the $j$-th covariance matrix $\bSigma_j$
is an $m_j\times m_j$-dimensional square matrix, $j=1,\ldots,k$. 
As \ctn{Hoff, manceur} suggest, a $k$-th ordered random tensor $\bSigma \in R^{ m_1 \times
  m_2...\times m_k}$ can be decomposed to a $k$-th ordered tensor $\bZ$
and $k$ number of covariance matrices $\bSigma_1,\ldots, \bSigma_k$ by
{Tucker product}, according to 
  $
   \bSigma = \bZ \times_1 \bSigma_1  \times_2 \bSigma_2 ...  \times_k  \bSigma_k$,
It can be proved that all tensors can be
decomposed into a set of covariance matrices \ctp{tucker}, though not uniquely. This may
cause difficulty in finding the correct combination of covariance matrices
that present the correlation structure of the data at hand. One way to solve
this problem is to use priors for the respective covariance parameters.}


We employ this likelihood in Equation~\ref{eqn:eqn1} to write the joint posterior
probability density of the mean tensor and covariance matrices, given the data. But prior to doing
that, we identify those
parameters--if any--that can be estimated in a pre-processing stage of the
inference, in order to reduce the computational burden of inference. Also,
it would be useful to find ways of
(kernel-based) parametrisation of the sought covariance matrices, thereby reducing
the number of parameters that we need to learn. 
To this effect, we undertake the estimation of the mean tensor is $\bM \in R^{ m_1 \times
  m_2...\times m_k}$. It is empirically estimated as the sample mean $\overline{\bv}$ of the sample
$\{\bv_1,\ldots,\bv_n\}$, s.t. $n$ repetitions of $\overline{\bv}$ form
the value $\bm$ of $\bM$. 
However, if necessary, the mean tensor itself can be regarded as a random variable and learnt from the data \ctp{chakrabarty2015bayesian}, 
The modelling of the covariance structure of this GP is discussed in the following subsection.


Ultimately, we want to predict the value of one variable, at which a
new or test data on the other variable is observed. 
\begin{proposition}
{
To perform inverse prediction of value $\bs^{(test)}$ of the input variable 
$\bS$, at which test data $\bv^{(test)}$ on $\bV$ is realised, we will
\begin{enumerate}
\item[---]sample from the posterior probability density of $\bs^{(test)}$
given the test data $\bv^{(test)}$, and (modal) values of the unknowns that
parametrise the covariance matrices the high-dimensional GP invoked to model
$\bxi(\cdot)$, subsequent to learning the marginals of each such unknown 
given the training data, using MCMC.
\item[---]sample from the joint posterior probability density
of $\bs^{(test)}$ and all other unknowns parameters of this 
high-dimensional GP, given training, as well as test data, using MCMC.
\end{enumerate}}
\end{proposition}
Computational speed of the first approach, is higher, as marginal
distributions of the GP parameters are learnt separately. When the training
data is small, or if the training data is not representative of the test data
at hand, the learning of $\bs^{(test)}$ via the second method may affect the
learning of the GP parameters.

\subsection{3 ways of learning covariance matrices}
\noindent
Let the $ij$-th element of $p$-th covariance matrix $\bSigma_p^{(m_p\times
  m_p)}$ be $\sigma_{ij}^{(p)}$; $j,i=1,\ldots,m_p$, $p\in\{1,\ldots,
k\}$. 
\begin{definition}
{At a given $p$, $\sigma_{ij}^{(p)}$ bears information about
covariance amongst the $i$-th and $j$-th slices of the $k$-th ordered data
tensor ${\bD}_{\bV}=(\bv_1,\ldots,\bv_{m_p})$,
s.t. $m_1\times\ldots\times m_{p-1}\times m_{p+1}\times\ldots\times
m_k$-dimensional $i$-th ``slice'' of data tensor
$\bD_{\bV}$ is measured value $\bv_i$ of $k-1$-th ordered tensor-valued $\bV$,
where the $i$-th slice is
realised at the $i$-th design point $\bs_i$. 
}
\end{definition}

The covariance between the $i$-th and $j$-th slices of data $\bD_{\bV}$
decreases as the slices get increasingly more disparate, i.e. with increasing
$\parallel \bs_i - \bs_j\parallel$. In fact, we can model $\sigma_{ij}^{(p)}$
as a decreasing function $K_p(\cdot,\cdot)$ of this disparity $\parallel\bs_i
- \bs_j \parallel$, where $K_p(\bs_i,\bs_j)$ is the covariance kernel
function, computed at the $i$-th and $j$-th values of input variable $\bS$. In
such a model, the number of distinct unknown parameters involved in the
learning of $\bSigma_p$ reduces from $m_p(m_p+1)/2$, to the
number of hyper-parameters that parametrise the kernel function
$K_p(\cdot,\cdot)$.

However, kernel parametrisation is not always possible. \\
--Firstly, this parametrisation may cause information loss and this may not be
acceptable\\ \ctp{aston}. \\
--Again, we will necessarily avoid kernel parametrisation, when we cannot find
input parameters, at which the corresponding slices in the data are realised.

In such situations, \\
--we can learn the elements of the covariance matrix directly using MCMC,
though direct learning of all distinct elements of $\bSigma_p$ is feasible, as
long as total number of all unknowns learnt by MCMC $\lesssim 200$.\\
--we can use an empirical estimation for the covariance matrix $\bSigma_p$.
We collapse each of the $m_p$ number of $k-1$-th ordered tensor-shaped slices
of the data, onto the $q$-th axis in the space ${\cal Y}$ of $\bV$, where we
can choose any one value of $q$ from $\{1,\ldots,k-1\}$. This will reduce each
slice to a $m_q$-dimensional vector, so that $\sigma_{ij}^{(p)}$ is covariance
computed using the $i$-th and $j$-th such $m_q$-dimensional vectors.

Indeed such an empirical estimate of any covariance matrix is easily
generated, but it indulges in linearisation amongst the different
dimensionalities of the observable $\bV$, causing loss of information about
the covariance structure amongst the components of these high-dimensional
slices.  This approach is inadequate when the sample size is small
because the sample-based estimate will tend to be incorrect; indeed
discontinuities and steep gradients in the data, especially in small-sample
and high-dimensional data, will render such estimates of the covariance
structure incorrect. Importantly, such an approach does not leave any scope
for identifying the smoothness in the function $\bxi(\cdot)$ that represents
the functional relationship between the input and output variables. Lastly,
the uncertainties in the estimated covariance structure of the GP remain
inadequately known.

\begin{proposition}
{
We model the covariance matrices as \\
--kernel parametrised,\\
--or empirically-estimated, \\
--or learnt directly using MCMC. }
\end{proposition}
An accompanying computational worry is the inversion of any of the covariance
matrices; for a covariance matrix that is an $m_p\times m_p$-dimensional
matrix, the computational order for matrix inversion is well known to be
${\cal O}(m^3_p)$ \ctp{FW92}.

\section{Kernel parametrisation}
\label{sec:kernel}
\begin{proposition}
{
Kernel parametrisation of a covariance matrix, when undertaken, uses
an Squared Exponential (SQE) covariance kernel
\begin{equation}
K(\bs_i, \bs_j) := A\left[\exp\left(-(\bs_i - \bs_j)^T \bQ^{-1}
    (\bs_i-\bs_j)\right)\right],\quad\forall i,j=1,\ldots,d, 
\label{eqn:kernel2}
\end{equation} 
where $\bQ$ is a diagonal matrix, the diagonal elements of which are the
length scale hyperparameters $\ell_1,\ldots,\ell_d\in{\mathbb R}_{>0}$ that tell us how quickly
correlation fades away in each of the $d$-directions in input space ${\cal
  X}$, s.t. the inverse matrix $\bQ^{-1}$ is also diagonal, with
the diagonal elements given as
${\displaystyle{\frac{1}{\ell_1},\ldots,\frac{1}{\ell_d}}}$, where
$q_c:=1/\ell_c$ is the smoothness hyperparameter along the $c$-th direction in
${\cal X}$, $c=1,\ldots,d$. We learn these $d$ unknown parameters from the
data.

Here $A$ is the global amplitude, that is subsumed as a scale factor, in one
of the other covariance matrices, distinct elements of which are learnt directly using MCMC.}
\end{proposition}

\begin{remark} 
{   
We avoid using a model for the kernel in which amplitude depends on the
locations at which covariance is computed, i.e. the model: 
$K(\bs_i, \bs_j) := a_{ij}\left[\exp\left(-(\bs_i - \bs_j)^T \bQ^{-1}
    (\bs_i-\bs_j)\right)\right]$, and use a model endowed with a global
amplitude $A$. This helps avoid learning a very large
number ($d(d+1)/2$) of amplitude parameters $a_{ij}$ 
directly from MCMC. }
\end{remark}
A loose interpretation of this amlitude modelling is that we have scaled all
local amplitudes $a_{ij}$ to be $\leq 1$ using the global factor $A$
($=\max\limits{ij}\{a_{ij}\}$), and these scaled local amplitudes are then
subsumed into the argument of the exponential in the RHS of the last equation,
s.t. the reciprocal of the correlation length scales, that are originally
interpreted as the elements of the diagonal matrix $\bQ^{-1}$, are now
interpreted as the smoothing parameters modulated by such local
amplitudes. This interpretation is loose, since the same smoothness parameter
cannot accommodate all (scaled by a global factor) local amplitudes$\in(0,1]$,
for all $\bs_i-\bs_j$.

\subsection{Including non-stationarity, by modelling hyperparameters of
  covariance kernels as realisations of Stochastic Process}
\noindent
By definition of the kernel function we choose, (Equation~\ref{eqn:kernel2}),
all functions $\bxi(\cdot)$ sampled from the tensor-variate GP, are endowed with
the same length scale hyperparameters $\ell_1,\ldots,\ell_d$, and global
amplitude $A$. However, the data on the output variable $\bV$ is not
continuous, i.e. similarity between $\bs_i$ and $\bs_j$ does not imply
similarity between $\bxi(\bs_i)$ and $\bxi(\bs_j)$, computed in a universal
way $\forall \bs_i,\bs_j\in{\cal X}$.  Indeed, then a stationary definition of
the correlation for all pairs of points in the function domain, is wrong.  One
way to generalise the model for the covariance kernel is to suggest that the
hyperparameters vary as random functions of the sample path.
\begin{theorem}
\label{th:1}
{
For ${\bV}=\bxi(\bS)$, with $\bS\in{\cal X}$ and $\bV\in{\cal Y}$, if the map $\bxi:{\cal
  X}\longrightarrow{\cal Y}$ is a Lipschitz-continuous map over the bound
set ${\cal X}\subseteq{\mathbb R}^d$, where absolute
value of correlation
between $\bxi(\bs_1)$ and $\bxi(\bs_2)$ is 
$$\vert corr(\bxi(\bs_1),\bxi(\bs_2))\vert := \displaystyle{K\left(
  \langle (\bs_1 -\bs_2),\bq\rangle^2 \right)}, \quad\forall 
\bs_1,\bs_2\in{\cal X},$$
with $$K(\bs_1,\bs_2) := \displaystyle{\exp\left[-{\langle (\bs_1
      -\bs_2),\bq\rangle^2}\right]},$$
then the vector $\bq$ of correlation hyperparameters is finite, and each element
of $\bq$ is
$\bxi$-dependent, i.e. 
$$\bq(\bxi) = (q_1(\bxi),\ldots,q_d(\bxi))^T\in{\mathbb R}^d. $$
}
\end{theorem}
\begin{proof}
{
For $\bS\in{\cal X}$, where ${\cal X}$ is a bounded subset of ${\mathbb R}^d$, and $\bV\in{\cal Y}$, the mapping $\bxi:{\cal
  X}\longrightarrow{\cal Y}$ is a defined to be Lipschitz-continuous map, i.e.
\begin{equation}
d_{\cal Y}(\bxi(\bs_1) - \bxi(\bs_2)) \leq L_{\bxi} d_{\cal
  X}(\bs_1,\bs_2),\quad \forall \bs_1,\bs_2\in{\cal X},
\label{eqn:tensor_lip}
\end{equation} 
--for constant $L_{\bxi}\in{\mathbb R}$, s.t. the infinum over all such
constants is the finite Lipschitz constant for $\bxi$;\\
--$({\cal X}, d_{\cal X})$ and $({\cal Y}, d_{\cal Y})$ are metric spaces.

Let metric $d_{\cal X}(\cdot,\cdot)$ be the $L_2$ norm: 
$$d_{\cal X}(\bs_1,\bs_2) := \parallel \bs_1-\bs_2\parallel, \quad\forall\bs_1,\bs_2\in{\cal X},$$
and the metric $d_{\cal Y}(\bxi(\cdot),\bxi(\cdot))$ be defined as (square
root of the
logarithm of) the inverse of
the correlation:
$$d_{\cal Y}(\bxi(\bs_1),\bxi(\bs_2)) := \sqrt{-\log\vert corr(\bxi(\bs_1),\bxi(\bs_2))\vert},\quad\forall \bs_1,\bs_2\in{\cal X},$$
--where correlation being a measure of affinity, 
$\log\vert 1/corr(\cdot,\cdot)\vert$, transforms this affinity into a squared distance for this
correlation model; so the transformation $\sqrt{\log\vert
  1/corr(\cdot,\cdot)\vert}$ to a metric is undertaken; \\
--and the given kernel-parametrised correlation is:
$$\vert corr(\bxi(\bs_1),\bxi(\bs_2))\vert := \exp[-{\langle (\bs_1-\bs_2),\bq\rangle}^2],
\quad\forall \bs_1,\bs_2\in{\cal X},\:\bq\in{\mathbb R}^d,$$
so that
$$d_{\cal Y}(\bxi(\bs_1),\bxi(\bs_2))={\langle
  (\bs_1-\bs_2),\bq\rangle}. $$ 
Then for the map $\bxi$ to be Lipschitz-continuous, we require:
\begin{equation}
\displaystyle{\sum\limits_{i=1}^d q_i^2(\bs_1^{(i)}-\bs_2^{(i)})^2} \leq
\displaystyle{L_{\bxi}^2\sum\limits_{i=1}^d (\bs_1^{(i)}-\bs_2^{(i)})^2},
\label{eqn:last}
\end{equation}
where the vector of correlation hyperparameters, $\bq=(q_1,\ldots, q_d)^T$, is
finite given finite $L_{\bxi}$.

By choosing to define 
\begin{equation}
q_{max}:= \max(q_1,\ldots,q_d),
\label{eqn:defn_qmax}
\end{equation}
and\\
$$(q_i^{'})^2:=\displaystyle{\left(\frac{q_i}{q_{max}}\right)^2}\leq 1,\forall i=1,\ldots,d,$$ 
inequation~\ref{eqn:last} is valid, if we choose the $\xi$-dependent, Lipschitz
constant $L_{\bxi}$ (that exists for this Lipschitz map) to be:
$$ L_{\bxi}^2 = q_{max}^2,$$
i.e. the map $\bxi$ is Lipschitz-continuous, if $q_{max}$ is
$\bxi$-dependent.\\
Then recalling definition, $q_{max}$ from Equation~\ref{eqn:defn_qmax}, it
follows that in general, $q_i$  is $\bxi$-dependent, $\forall i=1,\ldots, d$.
}
\end{proof}

Given discontinuities in the data on $\bV$, the function $\bxi(\cdot)$ is not
expected to obey the Lipschitz criterion defined in
inequation~\ref{eqn:tensor_lip} globally. We anticipate sample function
$\bxi(\cdot)$ to be locally or globally discontinuous.
\begin{lemma}
\label{lemma:1}
{Sample function $\bxi(\cdot)$ can be s.t.
\begin{enumerate}
\item[Case(I)] $\exists \bs_2\in{\cal X}$, s.t. $\nexists$ finite Lipschitz
  constant $L_{\bxi}^{(1,2)}>0$, for which $d_{\cal Y}(\bxi(\bs_1) - \bxi(\bs_2)) \leq
  L_{\bxi}^{(1,2)} d_{\cal X}(\bs_1,\bs_2)$. 
  Here the bounded set ${\cal X}\subset{\mathbb R}^d$.
\item[Case(II)] $\exists \bs_2, \bs_3\in{\cal X}$, with
$\parallel\bs_2-\bs_1\parallel\neq \parallel\bs_3-\bs_1\parallel$,
s.t. $d_{\cal Y}(\bxi(\bs_1) - \bxi(\bs_2)) \leq L_{\bxi}^{(1,2)} d_{\cal
  X}(\bs_1,\bs_2)$, but $d_{\cal Y}(\bxi(\bs_1) - \bxi(\bs_3)) \leq
L_{\bxi}^{(1,3)} d_{\cal X}(\bs_1,\bs_3)$; $L_{\bxi}^{(1,2)}\neq L_{\bxi}^{(1,3)}$. In such a case, the Lipschitz
constant used for the sample function $\bxi(\cdot)$ is defined to be 
$$L_{\bxi}=\max\{L_{\bxi}^{(i,j)}\}_{i\neq j; \bs_i,\bs_j\in{\cal X}}.$$

\end{enumerate}
If each function in the set $\{\bxi_1(\cdot),\ldots,\bxi_n(\cdot)\}$ is\\
--either globally Lipschitz, or is as described in Case~II, \\
--and Case~I does not hold true, then\\ 
$\forall \bs_1,\bs_2\in{\cal X},\:\:$ $\exists$ a finite $L_{max}>0$, where
$$L_{max}:= \max\limits_{\bxi}\{L_{\bxi_1}, L_{\bxi_2},\ldots, L_{\bxi_n}\},$$
where $L_{\bxi_i}$ is the $i$-th Lipschitz constant defined for the $i$-th
sample function $\bxi_i(\cdot)$, $i=1,\ldots,n$, \\i.e. $\exists$ a
finite Lipschitz constant for all $n$ sample functions.\\
 $\Longrightarrow \exists$ a universal correlation hyperparameter vector $\bq_{max}$ for
all $n$ sample functions (=$L_{max}$, by Theorem~\ref{th:1}).
}
\end{lemma}

\begin{lemma}
{
Following on from Lemma~\ref{lemma:1}, 
if for any $\bxi_i(\cdot)\in\{\bxi_1(\cdot),\ldots,\bxi_n(\cdot)\}$ 
Case~I holds, $\Longrightarrow$ finite maxima of
$\bxi_i(\cdot)\in\{\bxi_1(\cdot),\ldots,\bxi_n(\cdot)\}$ does not exist, \\$\Longrightarrow\nexists$ a finite Lipschitz constant $L_{max}$
for all $n$ sample functions, \\$\Longrightarrow\nexists$ a universal correlation
hyperparameter vector $\bq_{max}$, for all sample functions,\\
i.e. we need to model correlation hyperparameters to vary with the sample function.
}
\end{lemma}

\begin{remark}
Above, $q_1,\ldots,q_d$ are hyperparameters of the correlation kernel;
they are interpreted as the reciprocals of the length-scales
$\ell_1,\ldots,\ell_d$, i.e. $\ell_i=1/q_i, \forall i=1,\ldots,d$.
\end{remark}

\begin{remark}
\label{rem:2nd}
{If the map $\bxi:{\cal X}\longrightarrow{\cal Y}$ is Lipschitz-continuous,
(i.e. if hyperparameters $q_1,\ldots,q_d$ are $\bxi$-dependent, by
Theorem~\ref{th:1}), then by Kerkheim's Theorem \ctp{kerkheim}, $\bxi$ is
differentiable almost everywhere in ${\cal X}\subset{\mathbb R}^d$; this is a
generalisation of Rademacher's Theorem to metric differentials (see
Theorem 1.17 in \ctn{hazlas}). However, in our case, the function
$\bxi(\cdot)$ is not necessarily differentiable given discontinuities in the
data on the observable $\bV\in{\cal Y}$, and therefore, is not necessarily
Lipschitz. }
\end{remark} 

Thus, Theorem~\ref{th:1} and Lemma~\ref{lemma:1} negate
usage of a universal correlation length scale independent of sampled function
$\bxi(\cdot)$, in anticipation of discontinuities in the sample function.

\begin{proposition}
{
For ${\bV}=\bxi(\bS)$, with $\bS\in{\cal X}\subseteq{\mathbb R}^d$ and
$\bV\in{\cal Y}\subseteq{\mathbb R}^{(m_1\times\ldots\times m_k)}$, 
$$\vert corr(\bxi(\bs_1),\bxi(\bs_2))\vert := \displaystyle{\exp\left[-{\langle (\bs_1
      -\bs_2),\bq({\bxi})\rangle^2}\right]}, \quad\forall \bs_1,\bs_2\in{\cal X},$$
where $\bxi(\cdot)$ is a sample function of a tensor-variate GP.
Thus, in this updated model, $c$-th component $q_c=1/\ell_c$ of correlation
hyperparameter $\bq({\bxi})$ is
modelled as randomly varying with the sample function, $\bxi(\cdot)$, of the
tensor-variate GP, $\forall c=1,\ldots,d$.\\
In the Metropolis-within-Gibbs-based inference that we undertake, 
one sample function of the tensor-variate GP generated, in every iteration,
$\Longrightarrow q_c$ that we model above
$${\mbox{as randomly 
varying with the sample path of the tensor-variate GP,}}$$ 
$$\equiv{\mbox{is randomly varying with the iteration number variable}}\quad
T\in\{0,1,\ldots,t_{max}\}\subset{\mathbb Z}_{\geq 0},$$ 
$$\Longrightarrow{\mbox{We model}}\quad\ell_c = g_c(t),\quad c=1,\ldots,d,$$ where
this scalar-valued random function $g_c:\{0,1,\ldots,t_{max}\}\subset{\mathbb Z}_{>0}\longrightarrow
{\mathbb R}_{\geq 0}$, is modelled as a realisation from a
scalar-variate GP. }
\end{proposition}
Scalar-variate GP that $g_c(\cdot)$ is sampled
from, is independent of the GP that $g_{c'}(\cdot)$ is sampled from; $c\neq
c'; c,c'=1,\ldots,d$.
In addition, parameters that define the correlation function of the
generative scalar-variate GP can vary, namely the amplitude $A$ and scale
$\delta$ of one such GP might be different from another. Thus, 
scalar-valued functions sampled from GPs with varying correlation parameters
$A$ and $\delta$--even for the same $c$ value--should be marked by 
these descriptor variables $A>0$ and $\delta>0$. 
\begin{proposition}
{
We update the relationship between iteration number $T$ and correlation length
scale hyperparameter $\ell_c$ in the $c$-th direction in input space to be: $$\ell_c = g_{c,\bx}(t), \quad{\mbox{where
vector of descriptor variables is}}\quad \bX:=(A,\delta)^T, \quad{with}$$
--$A_c$ the amplitude variable of the SQE-looking covariance function of
the scalar-variate GP that $g_{c,\bx}(\cdot)$ is a realisation of.
$A_c$ takes the value $a_c\geq 0$; \\
--$\delta_c$ the length scale variable of the SQE-looking covariance
function of the scalar-variate GP that $g_{c,\bx}(\cdot)$ is a realisation of; 
$\delta_c \in{\mathbb R}_{> 0}$.\\
Then the scalar-variate
GPs that $g_{c,\bx}(\cdot)$ and $g_{c,\bx^{/}}(\cdot)$ are sampled from, have
distinct correlation functions if $\bx\neq  \bx^{/}$. Here $c=1,\ldots,d$.}
\end{proposition}

\begin{proposition}
{
Current value of correlation length scale hyperparameter
  $\ell_c$, acknowledges information on only the past $t_0$ number of
  iterations as in:
\begin{eqnarray}
\ell_{c} &=& g_{c,\bx}(t-t^{'}),\quad {\mbox{if}}\:\:t\geq t_0,\:\: c=1,\ldots,d;\: t^{'}=1,\ldots, t_0,\nonumber \\
\ell_{c} &=& \ell_c^{(const)},\quad {\mbox{if}}\:\:t =0,1,\ldots,t_0-1,\:\: \quad
c=1,\ldots,d,
\label{eqn:lenscl}
\end{eqnarray}
where $\ell_c^{(const)}$ is an unknown constant that we learn from the data,
during the first $t_0$ iterations. }
\end{proposition}  

As $g_{c,\bx}(t)$ is a realisation from a scalar-variate GP, the joint
probability distribution of $t_0$ number of values of the function
$g_{c,\bx}(t)$--at a given $\bx=(a, \delta)^T$--is Multivariate
Normal, with $t_0$-dimensional mean vector $\bM_{c,\bx}$ and $t_0\times
t_0$-dimensional covariance matrix $\bPsi_{c,\bx}$, i.e.
\begin{equation}
[g_{c,\bx}(t-1),\ldots, g_{c,\bx}(t-2), g_{c,\bx}(t-t_0)] \sim {\cal MN}(\bM_{c,\bx}, \bPsi_{c,\bx}).
\label{eqn:multvar}
\end{equation}
\begin{definition}
{Here $t_0$ is the number of iterations that we look back at, to collect the
dynamically-varying ``look back-data'' ${\bf D}_{c,t}^{(orig)}
:= \{\ell_{c,t-t_0},\ldots,\ell_{c,t-1}\}$ that is employed to learn
parameters of the scalar-variate GP that $g_{c,\bx}(\cdot)$ is modelled with.

\begin{enumerate}
\item[---]The mean vector $\bM_{c,\bx}$ is empirically estimated as the mean 
of the dynamically varying look back-data, s.t.
at the $t$-th iteration it is estimated as a $t_0$-dimensional vector with
each component
${\hat{m}}_{c,\bx}^{(t)}:=[\ell_{c,t-t_0}+\ldots+\ell_{c,t-1}]/t_0$.
\item[---]$t_0\times t_0$-dimensional covariance matrix is dependent on the
iteration-number and this is now acknowledged in the notation to state: $\bPsi_{c,\bx}(t)=\left[a_c\exp\left(
-\frac{(t_i-t_j)^2}{\delta_c^2}\right)\right],\:i,j=t-1,\ldots,t-t_0$.
\end{enumerate}
}
\end{definition}

In the $t$-th iteration, upon the empirical estimation of
the mean as given above, it is subtracted from the ``look back-data'' ${\bf D}_{c,t}^{(orig)}$
so that the subsequent
mean-subtracted look back-data is ${\bf D}_{c,t}
:= \{\ell_{c,t-t_0}-{\hat{m}}_{c,\bx}^{(t)},\ldots,\ell_{c,t-1}-{\hat{m}}_{c,\bx}^{(t)}\}$. It is
indeed this mean-subtracted sample that we use. 

\begin{definition}
In light of this declared usage of the mean-subtracted ``look back-data''
${\bf D}_{c,t}$, we update the likelihood over what is declared in Equation~\ref{eqn:multvar}, to
\begin{equation}
[g_{c,\bx}(t-1),\ldots, g_{c,\bx}(t-2), g_{c,\bx}(t-t_0)] \sim {\cal MN}(\b0,
\bPsi_{c,\bx}(t)),\quad \forall c=1,\ldots,d.
\label{eqn:multvar2}
\end{equation}
\end{definition}

\subsection{Temporally-evolving covariance matrix}
\label{sec:temporal}
\begin{theorem} { The dynamically varying covariance matrix of the
    Multivariate Normal likelihood in Equation~\ref{eqn:multvar2}, at
    iteration number $t\geq t_0$, is $$\bPsi_{c,\bx}(t) \sim
    \displaystyle{{\cal GWP}(d, \bG_c, k(\cdot,\cdot))},\quad where:$$
    the number of iterations we look back to is $t_0$; \\$k(\cdot,\cdot)$ is
    the covariance kernel parametrising the covariance function of the
    scalar-variate GP that generates the scalar-valued function
    $g_{c,\bx}(\cdot)$, at the vector $\bx=(a_c,\delta_c)^T$ of descriptor
    variables,
    s.t. $k(t_i,t_j)=\exp\left(-\frac{(t_i-t_j)^2}{\delta_c^2}\right),\:\forall
    t_i,t_j=t-1,\ldots,t-t_0$; \\$\bG_c$ is a positive definite square scale
    matrix $\bG_c$ of dimensionality $t_0$, containing the amplitudes of this
    covariance function; \\$c=1,\ldots,d$, with the space ${\cal X}$ of input
    variable $\bS$ $d$-dimensional.  }
\end{theorem}

\begin{proof}
{
The covariance
  kernel $k(\cdot,\cdot)$ that parametrises the covariance function of the
  scalar-variate GP
  that generates $g_{c,\bx}(t)$, is s.t. $k(t_i,t_i)$=1 $\forall i=1,\ldots,t_0$.

In a general model, at each iteration, a new value of the vector $\bx_c$ of
descriptor variables in the $c$-th direction in the space ${\cal X}$ of the
input variable $\bS$, is generated,
s.t. in the $t-t_i$-th iteration, it
is $\bx_{c,i}=(a_{c,i},\delta_{c,i})^T$; $t-t_i= t-1,\ldots,t-t_0$
$$\Longrightarrow{\mbox{at}}\quad T=t,\quad \{g_{c,\bx_1}(t),\ldots, g_{c,\bx_{t_0}}(t)\}
\quad{\mbox{is a sample of the random variable}}\quad g_{c,\bx}(t).$$ 
Now, $corr(g_{c,\bx}(t-t_i),g_{c,\bx}(t-t_j)) =
k(t_i,t_j)\delta({c,c^{/}})\delta({\bx,\bx^{/}})$, where $\delta({\cdot,\cdot})$ is the Delta function.\\
$\Longrightarrow$ sample estimate of $Cov(g_{c,\bx}(t-t_i),g_{c,\bx}(t-t_j))$ is 
$$Cov(g_{c,\bx}(t-t_i),g_{c,\bx}(t-t_j))=
\displaystyle{\sum\limits_{k=1}^{t_0} a_{c,i} a_{c,j} g_{c,\bx_k}(t-t_i) g_{c,\bx_k}(t-t_j)},\quad \forall t-t_i, t-t_j = t-1,\ldots,t-t_0,$$
is the $ij$-th element of matrix $\bPsi_{c,\bx}(t)$.\\
This definition of the covariance holds since mean of the
r.v. $g_{c,\bx}(t)$ is 0, as we have sampled the function from a zero-mean
scalar-variate GP. 

$${\mbox{Let}}\quad \bg_{c,\bx_k}(t) := (g_{c,\bx_k}(t-t_1),\ldots,
g_{c,\bx_{t_k}}(t-t_0))^T,\quad k=1,\ldots,t_0.$$

Let $\bG_c$ be a $t_0\times t_0$-dimensional diagonal matrix, the $i$-th diagonal element
of which is $a_{c,i}^2$. Then factorising the scale matrix $\bG_c=\bL_{G_c} \bL_{G_c}^T$,
$\bL_{G_c}$ is diagonal with the $i$-th diagonal element
$a_{c,i}$; $i=1,\ldots,t_0$. This is defined for every $c\in\{1,\ldots,d\}$.

Then at iteration number $T=t$, we define the current covariance matrix 
$$\bPsi_{c,\bx}(t) :=
\displaystyle{\sum\limits_{k=1}^{t_0} \bL_{G_c}\left(\bg_{c,\bx_k}(t)\right)^T \bg_{c,\bx_k}^T(t) \bL_{G_c}^T}.$$ 

Then 
$\bPsi_{c,\bx}(t)$ is distributed according to the Wishart distribution w.p, 
$\bG_c$ and $d$ \ctp{eaton}, i.e. the dynamically-varying covariance matrix is:
$$ \bPsi_{c,\bx}(t) \sim {\cal{GWP}}(d, \bG_c, k(\cdot,\cdot)).$$ 
}
\end{proof}
\begin{remark}
{If interest lies in learning the covariance matrix at any time point,
we could proceed to inference here from, in attempt of the learning of the
unknown parameters of this ${\cal{GWP}}$ process given the lookback-data ${\bf
  D}_{c,t}$. 

Our
learning scheme then would then involve compounding a Tensor-Variate GP
and a ${\cal{GWP}}$.}
\end{remark}
The above would be a delineated route to recover the temporal variation
in the correlation structure of time series data (as studied, for example by \ctn{wilson}).
\begin{remark}
{In our study, the focus is on high-dimensional data that display
discontinuities, and on learning the relationship $\bxi(\cdot)$ between the
observable $\bV$ that generates such data, and the system parameter
$\bS$--with the ulterior aim being parameter value prediction. So learning the
time-varying covariance matrix $\Psi(t)$ is not the focus of our method
development. }
\end{remark}

We want to learn $\bxi(\cdot)$ given training data ${\bf D}$.
The underlying motivation is to sample a new $g_{c,\bx}(\cdot)$ from 
a scalar-variate GP, at new values of
$a_1,\ldots,a_d,\delta_1,\ldots,\delta_d$, to subsequently sample a new tensor-valued
function $\bxi(\cdot)$, from the
tensor-normal GP, at a new value of its $d$-dimensional correlation
length scale hyperparameter vector $\bell$.

\subsection{2-layers suffice}
\label{sec:suffice}
\noindent
One immediate concern that can be raised is the reason for limiting the
layering of our learning scheme to only 2. It may be argued that just as
we ascribe stochasticity to the length scales $\ell_1,\ldots,\ell_d$ that
parametrise the correlation structure of the tensor-variate GP that models
$\bxi(\cdot)$, we need to do the same to the descriptor
variables $a,\:\delta$ that parametrise the correlation structure of the
scalar-variate GP that model $g_{c,\bx}(t)$. Following this argument, we would
need to hold $a,\:\delta$--or at least model the scale
$\delta$--to be dependent on the sample path of the scalar-variate GP,
i.e. set $\delta$ to be dependent on $g_{c,\bx}(\cdot)$.

However, we show below that a global choice of
$\delta$ is possible irrespective of the sampled function $g_{c,\bx}(\cdot)$, given that
$g_{c,\bx}:\{t-1,\ldots,t-t_0\}\subset{\mathbb Z}_{\geq 0}\longrightarrow{\mathbb R}_{\geq 0}$ is always
continuous (a standard result). In contrast, the function $\bxi(\cdot)$ not being
necessarily Lipschitz (see Remark~\ref{rem:2nd}), implies that the correlation
kernel hyperparameters $q_c$, are $\bxi$-dependent, $\forall c=1,\ldots,d$.

\begin{theorem}
\label{th:main1}
{
Given $\ell_c=g_{c,\bx}(t)$, with $T\in{\cal N}\subset{\mathbb Z}_{\geq 0}$ and
$\ell_c\in{\mathbb R}$, 
the map $g_{c,\bx}:{\mathbb Z}_{\geq
  0}\longrightarrow{\mathbb R}_{\geq 0}$ is a Lipschitz-continuous map,
$\forall c=1,\ldots,d$. Here ${\cal N}:=\{t-t_1,\ldots,t-t_0\}$}
\end{theorem}
The proof of this standard theorem is provided in Section~4 of the
supplementary Materials.

\begin{theorem}
  \label{th:main2} { For any sampled function $g_{c,\bx}:{\cal N}\longrightarrow{\mathbb R}_{\geq 0}$ realised from
    a scalar-variate GP that has a covariance function that is kernel-parametrised with an SQE kernel function, parametrised by
    amplitude and scale hyperparameters, the Lipschitz constant that defines
    the Lipschitz-continuity of $g_{c,\bx}(\cdot)$, is $g_{c,\bx}$-dependent,
    and is given by the reciprocal of the scale hyperparameter, s.t. the set of $t_0$
    values of scale hyperparameters, for each of the $t_0$ samples
    of $g_{c,\bx}(\cdot)$ taken from the scalar-variate GP, admits a finite minima.    }
\end{theorem}
\begin{proof} { For $\ell_c=g_{c,\bx}(T)$, $g_{c,\bx}:{\cal N}\subset{\mathbb
      Z}_{\geq 0}\longrightarrow{\cal G}\subset{\mathbb R}_{\geq 0}$ is a
    Lipschitz-continuous map, (Theorem~\ref{th:main1}), with $T\in{\cal N}$
    and $\ell_c\in{\cal G}$. (${\cal N}$ is defined in
    Theorem~\ref{th:main1}).  Distance between any $t-t_1,t-t_2\in{\cal N}$ is
    given by metric $$d_{\cal N}(t-t_1,t-t_2) := \vert t_1-t_2\vert.$$ Distance
    between $g_{c,\bx}(t-t_1)$ and $g_{c,\bx}(t-t_2)$ is given by metric $$d_{\cal
      G}(g_{c,\bx}(t-t_1),g_{c,\bx}(t-t_2)) := \sqrt{-\log\vert
      corr(g_{c,\bx}(t-t_1), g_{c,\bx}(t-t_2))\vert},$$ s.t. $d_{\cal
      G}(g_{c,\bx}(t-t_1),g_{c,\bx}(t-t_2))\geq 0$, and is finite (since
    $t-t_1,t-t_2$ live in a bound set, and $g_{c\bx}(\cdot)$ is continuous). 
The parametrised model of the correlation is
$$\vert corr(g_{c,\bx}(t-t_1), g_{c,\bx}(t-t_2))\vert := \displaystyle{K\left(
  \frac{(t_1 - t_2)^2}{\delta_g^2} \right)} \equiv
\displaystyle{\exp\left[-\frac{(t_1
      -t_2)^2}{\delta_g^2}\right]},$$
s.t. $\vert corr(g_{c,\bx}(t-t_1), g_{c,\bx}(t-t_2))\vert\in(0,1]$, where
$\delta_g > 0$ is the scale hyperparameter.

Now, Lipschitz-continuity of $g_{c,\bx}(\cdot)$ implies
\begin{equation}
d_{\cal G}(g_{c,\bx}(t-t_1),g_{c,\bx}(t-t_2)) \leq L_g d_{\cal N}(t-t_1,t-t_2),
\label{eqn:lips}
\end{equation}
where the Lipschitz constant $L_g$ is $g_{c,\bx}$-dependent (Theorem~\ref{th:1}). As $d_{\cal N}(t-t_1,t-t_2)\equiv \vert
t_1-t_2\vert \leq t_0$, where $t_0$ is a known finite integer, and as 
$d_{\cal G}(\cdot,\cdot)$ is defined as $\vert t_1-t_2\vert/\delta_g,\:\delta_g>0$ (using
definition of $d_{\cal G}(\cdot,\cdot)$), $L_g$ exists for $t_1, t_2$, and is
finite. We get 
\begin{equation}
L_g =\displaystyle{\frac{1}{\delta_g}}.
\label{eqn:lipdel}
\end{equation}
As $t-t_1,t-t_2$ is any point in ${\cal N}$, $L_g$ exists for all points in ${\cal N}$.

Let set $\bL:=\{L_{g_1},\ldots,L_{g_{t_0}}\}$, where $L_{g_i}$ 
defines the Lipschitz-continuity condition
(inequation~\ref{eqn:lips}) for the
$i$-th sample function $g_i(\cdot)$ from a scalar-variate GP.  
$$\exists
L_{max}:={\displaystyle{\max_{g}}}[\bL]={\displaystyle{\max_{g}}}\{L_{g_1},\ldots,L_{g_{t_0}}\},\quad{\mbox{where}}\quad
L_{max}>0\quad{\mbox{and is finite.}}$$ Thus, $L_{max}$ is a Lipschitz
constant that defines the Lipschitz continuity for any sampled function in
$\{g_{c,\bx}(t-t_1),\ldots,g_{c,\bx}(t-t_0)\}$, at any
iteraion number $t$ in a chain of finite and known number of iterations.\\
Then by Equation~\ref{eqn:lipdel}, $\exists \delta >0$, s.t.
$$\delta:={\displaystyle{\max_{g}}\left\{\frac{1}{\delta_{g_1}},\ldots,\frac{1}{\delta_{g_{t_0}}}\right\}}=
{\displaystyle{\min_{g}}\{{\delta_{g_1}},\ldots,{\delta_{g_{t_0}}}\}};\quad{\mbox{where}}\:\delta_{g_i}>0\forall i=1,\ldots,t_0.$$
Here $L_{g_i} =\displaystyle{\frac{1}{\delta_{g_i}}}; \: i=1,\ldots,t_0$.
}
\end{proof}
 
\begin{theorem}
{
Given $\ell_c=g_{c,\bx}(t)$, where $g_{c,\bx}:{\cal N}\longrightarrow{\cal G}$
is a Lipschitz-continuous function, sampled from a scalar-variate GP, the
covariance function of which, computed at any 2 points $t-t_1,t-t_2$ in the input
space ${\cal N}$, is kernel parametrised as
$$Cov(t_1,t_2) = \displaystyle{a_c K\left(
  \frac{(t_1 - t_2)^2}{\delta_c^2} \right)} \equiv
\displaystyle{a_c \left(\exp\left[-\frac{(t_1
      -t_2)^2}{\delta_c^2}\right]\right)},$$
where ($a_c$, the amplitude hyperparameter and) the scale hyperparameter of this kernel is $\delta_c$ that is independent of the sample function $g_{c,\bx}(\cdot)$; $c\in\{1,\ldots,d\}$.}
\end{theorem}

\begin{proof}
{
By Theorem~\ref{th:main2},
$\delta_c:={\displaystyle{\min_{g_c}}\{{\delta_{g_{c,1}}},\ldots,{\delta_{g_{c,n}}}\}}$
exists for any $c\in\{1,\ldots, d\}$. Then the scalar-variate GP that models the sample function 
$g_{c,\bx}(\cdot)$ has a covariance kernel that is marked by the finite 
scale hyperparameter $\delta_c$, independemt of the sample function.}
\end{proof}

\begin{remark}
\label{rem:distinguish}
{
That a stationary scale hyperparameter $\delta$ that is independent of the
sample path can define the covariance kernel of the scalar-variate GP that
$g_{c,\bx}(\cdot)$ is sampled from, owes to the fact that any such sample
function $g_{c,\bx}(\cdot)$ is continuous given the nature of the map (from
a subset of integers to reals). However, when the sample function from a GP is
not continuous, (such as $\bxi(\cdot)$ that is modelled with the
tensor-variate GP discussed above), a set of values of the sample function-dependent 
scale  hyperparameter(s) of the covariance kernel of the corresponding GP,
will not admit a minima, and therefore, in such a case, a global scale 
hyperparameter cannot be ascribed to the covariance kernel of the generating
GP. This is why we need to retain the correlation length scale hyperparameter
$\ell_c$ to be dependent on the tensor-valued sample function $\bxi(\cdot)$,
but the scale hyperparameter $\delta_c$ is no longer dependent on the
scalar-valued sample function $g_{c,\bx}(\cdot)$. In other words, we do not
require to add any further layers to our learning strategy, than the two
layers discussed.}
\end{remark}

\subsection{Learning using a Compound Tensor-Variate $\&$ Scalar-Variate GPs}
\noindent
We find inference defined by a sequential sampling from the 
scalar-variate GPs (for each of the $d$ directions of input
space), followed by that from tensor-variate GP, directly relevant to our
interests. Thus our learning involves a Compound
tensor-variate and multiple scalar-variate GPs. To abbreviate,
we will refer below to such a Compound Stochastic Process, as a ``$nested-GP$'' model.

\begin{remark}
{
As $\delta_c, a_c$ are not stochastic, hereon, we absorb the dependence of the function $g(\cdot)$ on the direction
index, via the descriptor parameters, and refer to this function as
$g_{\bx_c}(t)$; $c=1,\ldots,d$.}
\end{remark}

\begin{definition} {
    $Nested-GP$ model:\\
    for $\bV=\bxi(\bS)$, 
    $$\bxi(\cdot)\sim{\mbox{tensor-variate GP}},$$ s.t.  joint
    probability of $n$ observations of $k-1$-th ordered tensor-valued variable
    $\bV$ (that comprise training data ${\bf D}$), is $k$-th ordered Tensor Normal, with $k$ covariance matrices--which
    are empirically estimated, or learnt directly using MCMC, or kernel
    parametrised, s.t. length scale parameter $\ell_1,\ldots,\ell_d$ of this
    covariance kernel, is each modelled as a dynamically varying function
    $\ell_c=g_{\bx_c}(t)$, where
$$g_{\bx_c}(t)\sim\:c-{\mbox{th scalar-variate GP}},$$ 
$\Longrightarrow$joint probability of the last $t_0$ observations of $\ell_c$ (that
comprise ``lookback data'' ${\bf D}_{c,t}$), is
Multivariate Normal, the covariance function of which is parametrised by a
kernel indexed by the $c$-th, stationary descriptor parameter
    vector $\bx_c=(a_c,\delta_c)^T$, where $a_c$ is the amplitude and
    $\delta_c$ the scale-length hyperparameter of the SQE-looking covariance
    kernel; $c=1,\ldots,d$.}
\end{definition}

\begin{definition} {
    $Nonnested-GP$ model:\\
    for $\bV=\bxi(\bS)$, 
    $$\bxi(\cdot)\sim{\mbox{tensor-variate GP}},$$ s.t. joint
    probability of observations of 
    $\bV$ is $k$-th ordered Tensor Normal, with $k$ covariance matrices--which
    are empirically estimated, or learnt directly using MCMC, or kernel
    parametrised, s.t. length scale parameter $\ell_1,\ldots,\ell_d$ of this
    covariance kernel, is each treated as a stationary unknown. All learning
    is undertaken using training data ${\bf D}$.}
\end{definition}


\section{Inference}
\label{sec:inference}
\noindent
We undertake  inference with Metropolis-within-Gibbs.
Below $\theta^{(t\star)}$ indicates proposed value of parameter $\theta$ in the $t$-th
iteration, while $\theta^{(t)}$ refers 
to the value current in the $t$-th iteration.

\begin{itemize}
\item $Nested-GP$:
\begin{enumerate}
\item In $t>t_0$-th iteration, propose amplitude and
  scale-length of $c$-th scalar-variate GP as: 
$$ a_c^{(t\star)} \sim {\cal TN}(a_c^{(t-1)}, 0, v_a^{(c)}),\quad\forall
c=1,\ldots,d, $$
$$ \delta_c^{(t\star)} \sim {\cal N}(\delta_c^{(t-1)}, 0,
v_\delta^{(c)}),\quad\forall c=1,\ldots,d,$$ where ${\cal N}(\cdot)$ is
Normal, and ${\cal TN}(\cdot,0,\cdot)$ is
a Truncated Normal density left-truncated at 0. \\$v_a^{(c)},
v_\delta^{(c)}$ refer to constant, experimentally-chosen variances.
\item As length scale hyperparameter $\ell_c=g_{\bx_c}(t)\sim GP(0,
  \exp\left(-(\cdot-\cdot)^2/ 2\delta_c^2\right))$, 
probability of
the current lookback data ${\bf D}_{c,t}$ given parameters of this $c$-th
scalar-variate GP, is Multivariate Normal with mean vector $\bzero$ and a
current covariance matrix 
$\bPsi_c^{(t-1)}:=\left[ a_c^{(t-1)} \exp\left(-\frac{(t_i
      - t_j)^2}{2(\delta_c^{(t-1)})^2}\right)\right]; \quad
t_i,t_j=t-1,\ldots,t-t_0.$ Similarly, the likelihood of the proposed
parameters
can be defined. These enter computation of the acceptance ratio in the first
block of Metropolis-within-Gibbs.
\item At the updated parameters $\delta_c,a_c$, at $T=t$, length scale hyperparameters $\ell_1,\ldots,\ell_d$ are
  rendered Normal variates s.t. $$\ell_c^{t\star}\sim{\cal N}(\ell_c^{(t-1)},
  a_c^{(t\star)}),$$ under a Random Walk paradigm, when the mean of this
  Gaussian distribution is the current value of the $\ell_c$ parameter; $\forall c=1,\ldots,d$.
\item The proposed and current values of $\ell_1,\ldots,\ell_d$ inform on the
  acceptance ratio in the 2nd block of our inference, along with other,
  directly learnt parameters, of the covariance structure of the tensor-variate
  GP that $\bx(\cdot)$ is sampled from.
\end{enumerate}

\item $Nonnested-GP$:
\begin{enumerate} 
\item In the first block of Metropolis-within-Gibbs, $\ell_1,\ldots,\ell_d$
  are updated, once proposed as Normal variates, with experimentally chosen
  constant variance of the respective proposal density. 
\item Updating of directly-learnt elements of relevant covariance matrices is
  undertaken in the 2nd block, and the acceptance ratio that invokes the
  tensor-normal likelihood, is computed to accept/reject these proposed
  values, at the $\ell_c$ variable values that are updated in the first block
  of Metropolis-within-Gibbs.
\end{enumerate}
\end{itemize}

Details on inference is presented in Section~1 of the Supplementary
Materials.

\section{Application}
\label{sec:application}
\noindent
We illustrate our method using an application on astronomical data.  In this
application, we are going to learn the location of the Sun in the 
Milky Way modelled as a 2-dimensional disk. 
The training data ${\bf D}$ is cuboidally-shaped, and is of
dimensionalities $m_1\times m_2\times m_3$, where $m_1=2, m_2=50, m_3\equiv
n=216$, i.e. the 3-rd ordered tensor $\bD_{\bV}$ comprises of $n=216$ matrices
of dimension $50\times 2$, where $i$-th value of the matrix-variate observable
$\bV^{(50\times 2)}$ is realised at $i$-th value of system parameter vector $\bS$, s.t.
${\bf D}=\{(\bs_i,\bv_i)\}_{i=1}^n$. The 3rd-ordered tensor $\bD_{\bV}^{(m_1\times m_2\times n)}=(\bv_1,\vdots,\ldots,\vdots\bv_n)$ 

The training data comprises the $m_1=$2-dimensional velocity vectors of a
sample of $m_2=50$ number of real stars that exist around the Sun, in a model
Milky Way disk, where the matrix-variate r.v. $\bV^{(50\times 2)}$ comprising
such velocity vectors of this chosen stellar sample, are generated via
numerical simulations conducted with $n=216$ different astronomical models of
the Galaxy, with each such model of the Galaxy distinguished by a value of the
Milky Way feature parameter vector $\bS\in{\mathbb R}^d$, $d$=2
\ctp{dc2007}. Thus, $\bV=\bv_i$ at the $i$-th design point $\bs_i$,
$i=1,\ldots,216$. As $\bV$ is affected by $\bS$, we write $\bV=\bxi(\bS)$, and
aim to learn the high-dimensional function $\bxi(\cdot)$, with the aim of
predicting value of either $\bV$ or $\bS$, at a given value of the other.

In particular, there exists the test data $\bv^{(test)}$ that comprises the
$m_1=2$-dimensional velocity vectors of the 50 identified, stellar neighbours
of the Sun, as measured by the Hipparcos satellite \ctp{dc2007}. It is the
same 50 stars for which velocity vectors are simulated at each design
point. However, we do not know the real Milky Way feature parameter vector
$\bs^{(test)}$ at which $\bV=\bv^{(test)}$ is realised.

Since we are observing velocities of stars around the Sun, the observed
velocities will be impacted by certain Galactic features. These features
include location of the Sun $\bS$. Thus, the observed matrix $\bv^{(test)}$,
can be regarded as resulting from the Galactic features (including the sought
solar location) to bear certain values. So, fixing all Galactic features other
than the location $\bS$ of the Sun in the simulations that generate the
training data, the matrix $\bV$ of stellar velocities is related to $\bS$,
i.e.  $\bV=\bxi(\bS)$. The input variable $\bS$ is then also the location
from which an observer on Earth (or equivalently the Sun, on Galactic length
scales), observes the 2-dimensional velocity vectors of $m_2$
(=50) of our stellar neighbours.

\ctn{dc2007} generated the training data by first placing a regular
2-dimensional polar grid on a chosen annulus in an 2-dimensional astronomical
model of the MW disk. In the centroid of each grid cell, an observer was
placed. There were $n$ grid cells, so, there were $n$ observers placed in this
grid, such that the $i$-th observer measures velocities of $m_{2i}$ stars
that land in her grid cell, at the end of a simulated evolution of a sample
of stars that are evolved in this model of the MW disk, under the influence
of the feature parameters that mark this MW model. We indexed the $m_{2i}$
stars by their location with respect to the observer inside the grid cell, and
took a sample of $m_2=50$ stars from this collection of $m_{2i}$ stars;
$i=1,\ldots,n=216$. Thus, each observer records a matrix (or sheet) of 2-dimensional velocity vectors of $m_2$ stars. The test data measured by the Hipparcos satellite is
then the 217-th sheet, except we are not aware of the value of $\bS$ that this
sheet is realised at.

The solar location vector is 2-dimensional, i.e. $d$=2 since the Milky Way
disk is assumed to be 2-dimensional, i.e. $\bS=(S_1,S_2)^T$, s.t in this polar
grid, $S_1$ tells us about the radial distance between the Galactic centre and
the observer, while $S_2$ denotes the angular location of the observer in the
MW disk, w.r.t. a pre-fixed axis in the MW, namely, long axis of an elongated
bar of stars that lies pivoted at the Galactic centre, as per the astronomical
model of the MW that was used to generate the training data.

In \ctn{chakrabarty2015bayesian}, the matrix of velocities was vectorised, so
that the observable was then a vector. In our case, the observable is $\bV$--a
matrix. The process of vectorisation, causes \ctn{chakrabarty2015bayesian} to
undergo loss of correlation infomation. Our work allows for clear
quantification of such covariances. More importantly, our work provides a
clear template for implementing methodology for learning given
high-dimensional data that comprise measurements of a tensor-valued
observable. As mentioned above, the empirical estimate of the mean tensor is
obtained, and used as the mean of the Tensor Normal density that represents
the likelihood.

To learn $\bxi(\cdot)$, we model it as a realisation from a high-dimensional
GP, s.t, joint of $n$ values of $\bxi(\cdot)$--computed at
$\bs_1,\ldots,\bs_n$--is 3-rd Tensor Normal, with3 covariance matrices:
that inform on:\\
--amongst-observer-location covariance ($\bSigma_3^{(216\times216)}$),\\
--amongst-stars-at-different-relative-position-w.r.t.-observer covariance ($\bSigma_2^{(50\times 50)}$), and \\
--amongst-velocity-component covariance ($\bSigma_1^{(2\times 2)}$).

The elements of $\bSigma_2$ are not learnt by MCMC.\\
--Firstly, there is no input space variable that can be identified, at which
the $ij$-th element of $\bSigma_2$ can be considered to be realised;
$i,j=1,\ldots,50$, where this $ij$-th element gives the covariance amongst
the $i$-th and $j$-th, $216\times 2$-dimensional matrices within the 3-rd
ordered tensor $\bD_{\bV}$.  Effectively, the 41st star could have been
referred to as the 3rd star in this stellar sample, and the vice versa, i.e. 
there is no meaningful ordering in the labelling of the sampled stars with 
these indices. Therefore, we cannot use these labels as values of an input 
space variable, in terms of which, the covariance between the $i$-th and $j$-th
$216\times 2$-dimensional velocity matrices can be kernel-parametrised. \\ 
--Secondly, direct learning of the 50(51)/2 distinct elements of $\bSigma_2$,
using MCMC, is ruled out, given that this is a large number. \\
--In light of this, we will perform empirical estimation of
$\bSigma_2$. 
\begin{definition}
{Covariance between the $216\times 2$-dimensional stellar velocity matrix
$\bW_i:=[v^{(i)}_{pq}]$ of the sampled star labelled by index $i$, and the
matrix $\bW_j:=[v^{(j)}_{pq}]$
of the star labelled as $j$, ($p=1,\ldots,216; q=1,2$), is estimated as ${\widehat{\sigma_{ij}^{(2)}}}$, where:\\
${\widehat{\sigma_{ij}^{(2)}}}=$
$$ \displaystyle{
\frac{1}{2-1} \times 
\sum_{q=1}^2 \left[
                   \frac{1}{216} \times 
\left(\sum_{p=1}^{216} (v^{(i)}_{pq} - \bar{v}^{(i)}_q)
                       \times
                       (v^{(j)}_{pq} - \bar{v}^{(j)}_q)
\right)\right]},$$
where $\bar{v}^{(i)}_q=\displaystyle{\frac{\left(\sum_{p=1}^{216} v^{(i)}_{pq}\right)}{216}}$ is the sample mean of the $q$-th column of the 
matrix $\bV_i=[v^{(i)}_{pq}]$. }
\end{definition}

The 3 distinct elements of the $2\times 2$-dimensional covariance matrix
$\bSigma_1$ are learnt directly from MCMC. These include the 2 diagonal
elements $\sigma_{11}^{(1)}$, $\sigma_{22}^{(1)}$ and
$\rho:=\displaystyle{\frac{\sigma^{(1)}_{12}}{\sqrt{\sigma^{(1)}_{11}\sigma^{(1)}_{22}}}}$

We perform kernel parametrisation of $\bSigma_3$, using the SQE kernel such
that the $jp$-th element of $\bSigma_3$ is kernel-parametrised as
$[\sigma_{jp}] = \displaystyle{\exp\left(-(\bs_j-\bs_p)^T \bQ^{-1} (\bs_j-\bs_p)\right)}, j,p=1,\ldots,216.$
Since $\bS$ is a 2-dimensional vector, $\bQ$ is a 2$\times$ 2 square diagonal
matrix, the elements $\ell_{1}, \ell_{2}$ of which, represent the
the correlation length scales.

Then in the ``$nonnested-GP$'' model, we learn the (modelled as stationary)
$\ell_1, \ell_2$, along with $\sigma_{11}^{(1)}$, $\sigma_{22}^{(1)}$ and
$\rho$. 

Under the $nested-GP$ model, $\ell_c$ is modelled as $\ell_c=g_{\bx_c}(t)$,
where at iteration number $T=t$ $g_{\bx_c}(t)$ is sampled from the $c$-th
zero-mean, scalar variate GP, amplitude $a_c$ and correlation length scale 
$\delta_c$ of which we learn, for $c=1,2$, in addition to the parameters
$\sigma_{11}^{(1)}$, $\sigma_{22}^{(1)}$ and
$\rho$.
  
The likelihood of the training data given the covariance matrices of the
tensor-variate GP, is then given as per Equation~\ref{eqn:eqn1}:
  \begin{equation}
  \begin{aligned}
&{\cal L}({\bf D}\vert
\ell_1,\ell_2,\sigma_{11}^{(1)},\sigma_{22}^{(1)},\rho)=(2\pi)^{-m/2}(\prod_{i=1}^{3}|\bSigma_i|^{-m/2m_i})
\\ &\times \exp(-\Vert ({\bD}_{\bV}-\hat{\bM})\times_1 {\bA_1}^{-1}\times_2 {\hat{\bA_2}}^{-1}  \times_3 \bA_3^{-1} \Vert^2/2).
  \label{eqn:eqn3_bef}
  \end{aligned}
\end{equation}
where $\bSigma_p = \bA_p \bA^{T}_p$, $p=1,2,3$ and ${\hat{\bM}}$ is the
empirical estimate of the mean tensor and $\hat{\bSigma_2}$ is the empirical
estimate of the covariance matrix $\bSigma_2$ such that ${\hat{\bSigma_2}} =
{\hat{\bA_2}} {\hat{\bA_2}}^{T}$. Here $m_3=216$, $m_2=50$, $m_1=2$, and
$m=m_1 m_2 m_3$.  One or more of the covariance matrices is kernel
parametrised, where the kernel is a function of pairs of values of the input
variable $\bS$--this explains the dependence of the RHS of this equation on
the whole of ${\bf D}$, with the data tensor $\bD_{\bV}$ contributing partly
to training data ${\bf D}$.

This allows us to write the joint posterior probability density of the unknown
parameters given training data ${\bf D}$. We generate posterior samples from
it using Metropolis-within-Gibbs. To write this posterior, we impose
non-informative priors $\pi_0(\cdot)$ on each of our unknowns (Gaussian with
wide, experimentally chosen variances, and mean that is the arbitrarily chosen
seed value of $\ell_{\cdot}$; Jeffry's priors on $\bSigma_1$). The posterior
probability density of our unknown GP parameters, given the training data is
then
\begin{equation}
\begin{aligned}
\pi(\ell_{1}, \ell_{2}, \sigma_{11}^{(1)},\sigma_{22}^{(1)}, \rho\vert{\bf D})
\propto  {\cal L}({\bD}_{\bV} \vert \bSigma_1,\bSigma_3)\times \pi_0(\ell_{1}) \pi_0(\ell_{2}) \pi_0(\bSigma_1).
  \end{aligned}
\label{eqn:marginal_bef}
\end{equation}
The results of our learning and estimation of the mean and covariance
structure of the GP used to model this tensor-valued data, is discussed below
in Section~\ref{sec:results}. 

\begin{definition}
{
The joint posterior probability density of the
unknown parameters given 
the training data ${\bf D}$ that comprises the velocity tensor $\bD_{\bV}$, under the $nested-GP$
model is given by
\begin{equation}
\begin{aligned}
&\pi(\delta_{1}, \delta_{2}, a_1, a_2. \ell_1,\ell_2,
\sigma_{11}^{(1)},\sigma_{22}^{(1)}, \rho\vert {\bf D}) \propto
(2\pi)^{-m/2} \left(\prod_{i=1}^{3}|\bSigma_i|^{-m/2m_i}\right)\\ 
&\times \exp(-\Vert ({\bD}_{\bV}-\hat{\bM})\times_1 {\bA_1}^{-1}\times_2
{\hat{\bA_2}}^{-1}  \times_3 \bA_3^{-1} \Vert^2/2)\times\\
&\displaystyle{\prod\limits_{c=1}^2 \frac{1}{\sqrt{\det(2\pi\bPsi_{\bx_c})}}
\exp\left[-\frac{1}{2}
    (\bell_c^{(t_0)})^T \left(\bPsi_{\bx_c}\right)^{-1}(\bell_c^{(t_0)})\right]}
\times \pi_0(\bSigma_1),
  \end{aligned}
\label{eqn:marginal_bef2}
\end{equation}
where $\bell_c^{(t_0)}:=(\ell_c^{(t-t_0)},\ldots,\ell_c^{(t-1)})^T$, and
$ij$-th element of the covariance matrix $\bPsi_{\bx_c}$ is
$\displaystyle{\left[{a_c\exp\left[-\frac{(t_i-t_j)^2}{2(\delta_c)^2}\right]}\right]}$,
$i,j=1,\ldots,t_0$. N.B. the $t$-dependence of the covariance matrix
$\bPsi_{\bx_c}$ is effectively suppressed, given that this dependence comes in
the form $t-t_i -(t-t_j)$.
}
\end{definition}
We generate posterior
samples using MCMC, to identify the marginal posterior probability
distribution of each unknown. The marginal then allows for the computation of
the 95$\%$ HPD. 

\section{Inverse Prediction--2 Ways}
\label{sec:prediction}
\noindent
We aim to predict the location vector $\bs^{(test)}$ of the Sun in the Milky
Way disk, at which real (test) data
$\bv^{(test)}$ on the 2-dimensional velocity vectors of 50 identified stellar
neighbours of the Sun, measured by the {\it Hipparcos} satellite. We undertake
this, subsequent to learning of relation $\bxi(\cdot)$ between solar location
variable $\bS$ and stellar velocity matrix-valued variable $\bV$, using
astronomically-simulated (training data). 
\begin{definition}
{The tensor that includes both test and training data has dimensions of $217\times 50\times 2$. We call this augmented data
$\bD^*=\{\bv_1,...,\bv_{50},\bv^{(test)}\}$, to distinguish it from the
tensor $\bD_{\bV}$ that lives in the training data. Here $\bv_i$ is realised
at design point $\bs_i$, but the $\bs^{(test)}$ at which $\bv^{(test)}$ is
realised, is not known.}
\end{definition}

\begin{remark}
{This 217-th sheet of (test) data is realised at the unknown value
$\bs^{(test)}$ of $\bS$, and upon its inclusion, the updated covariance
amongst the sheets generated at the different values of $\bS$, is renamed
$\bSigma_1^*$, which is now rendered $217\times 217$-dimensional. Then
$\bSigma_1^*$ includes information about $\bs^{(test)}$ via the
kernel-parametrised covariance matrix $\bSigma_3$. The effect of inclusion of
the test data on the other covariance matrices is less; we refer to them as
(empirically estimated) ${\hat{\bSigma_2^*}}$ and $\bSigma_3^*$. The updated
(empirically estimated) mean tensor is ${\hat{\bM}}^*$. }
\end{remark}

The likelihood for the augmented data is:
  \begin{equation}
  \begin{aligned}
{\cal L}(\bD^*|\bs^{(test)}, \bSigma_1^*,\bSigma_3^*) =&
\displaystyle{(2\pi)^{-m/2}\left(\prod\limits_{i=1}^{3}|\bSigma_i^*|^{-m/2m_i}\right)}\times \\
&{\displaystyle{\exp\left[-\Vert (\bD^*-{\hat{\bM}}^*)\times_1 ({\bA_1^*})^{-1} \times_2 ({\hat{\bA_2^*}})^{-1} \times_3 ({\bA_3^*})^{-1} \Vert^2/2\right]}}
  \label{eqn:eqn4}
  \end{aligned}
\end{equation}
where ${\hat{\bA_2^*}}$ is the square root of ${\hat{\bSigma_2^*}}$.
Here $m_1=217$, $m_2=50$, $m_3=2$, and $m=m_1
m_2 m_3$. Here $\bA_1^*$ is the square root of $\bSigma_1^*$ and
depends on $\bs^{(test)}$.

The posterior of the unknowns given the test+training data is:
\begin{equation}
\begin{aligned}
\pi(s_1^{(test)},s_2^{(test)},\bSigma_1^*,\bSigma_3^*\vert \bD^*) \propto &
{\cal L}(\bD^*|s_1^{(test)},s_2^{(test)},\bSigma_1^*,\bSigma_3^*)\times\\
& \pi_0(s_1^{(test)})\pi_0(s_2^{(test)})\pi_0(q_{2}^{(*)})\pi_0(q_{1}^{(*)}) \pi_0(\bSigma_3^*).
  \end{aligned}
\label{eqn:marginal}
\end{equation}
\begin{remark}
{We use $\pi_0(s_p^{(test)})={\cal U}(l_p, u_p),\:p=1,2$, where $l_p$ and $u_p$
are chosen depending on the spatial boundaries of the fixed area of the Milky
Way disk that was used in the astronomical simulations by
\ctn{dc2007}. Recalling that the observer is located in a two-dimensional
polar grid, \ctn{dc2007} set the lower boundary on the value of the angular
position of the observer to 0 and the upper boundary is $\pi/2$ radians,
i.e. 90 degrees, where the observer's angular coordinate is the angle made by
the observer-Galactic centre line to a chosen line in the MW disk. The
observer's radial location is maintained within the interval [1.7, 2.3] in
model units, where the model units for length are related to galactic unit for
length, as discussed in Section~\ref{sec:astro}.}
\end{remark}

In the second method for prediction, we
infer $\bs^{(test)}$ by
sampling from the posterior of $\bs^{(test)}$ given the
test data and the modal values of the parameters
$q_{1}, q_{2}, \sigma_{11}^{(1)},
\rho,\sigma_{22}^{(1)}$ that were learnt using the training data.
Let modal value of $\bSigma_3$, learnt using ${\bf D}$ be
$[(\sigma_3^{(M)})_{jp}]_{j=1;p=1}^{217,217}$, 
Similarly, the modal value $\bSigma_1^{(M)}$ that was learnt using the
training data, is used. 
The posterior of $\bs^{(test)}$, at learnt (modal) values is then
\begin{equation}
\begin{aligned}
&\pi(s_1^{(test)},s_2^{(test)}\vert \bD^*,\bSigma_1^{(M)},\bSigma_3^{\star}) \propto \\
&{\cal L}(\bD^*|s_1^{(test)},s_2^{(test)},\bSigma_1^{(M)},\bSigma_3^{\star})\times \pi_0(s_1^{(test)})\pi_0(s_2^{(test)})
\times \pi_0(q_{2}^{(M)})\pi_0(q_{1}^{(M)}) \pi_0(\bSigma_3)|\bV^*).
  \end{aligned}
\label{eqn:marginalpred}
\end{equation}
where ${\cal L}(\bD^*|s_1^{(test)},s_2^{(test)},\bSigma_1^*,\bSigma_3^{(M)})$ is as given in Equation~\ref{eqn:eqn3_bef}, with $\bSigma_3$ replaced by $\bSigma_3^*$, and $\bSigma_1$ replaced by its modal value $\bsigma_1^{(M)}$. The priors on $s^{(test)}_1$ and $s^{(test)}_2$ are as discussed above.
For all parameters, we use Normal proposal densities that have experimentally chosen variances.

\begin{figure}[!t]
     \begin{center}
  {
       \includegraphics[width=10cm]{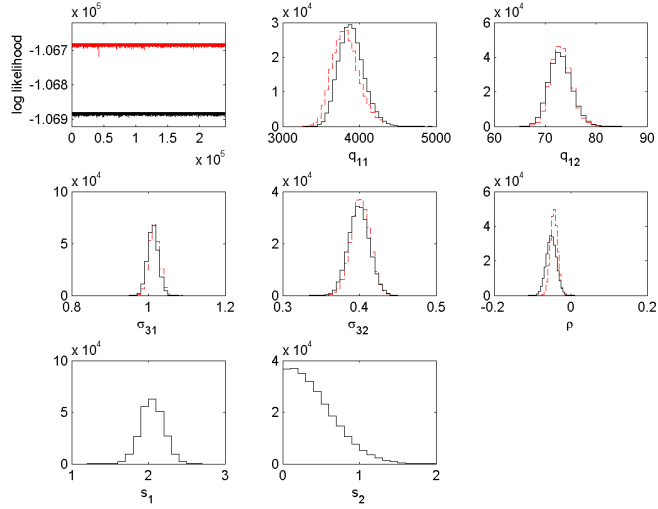}
  }
     \end{center}
     \caption{Results from run done with training data ${\bf D}$ with the
       $nonnested-GP$ model, are shown in grey (or red in the
       electronic version), while results from run undertaken with
       training and test data, $\bD^{\star}$, in this $nonnested-GP$ model,
       are depicted in black. Traces of the logarithm of the likelihood are
       displayed from the two runs in the top left panel. Reciprocal of the
       length scale parameters are the shown in the top middle and right
       panels; here $q_c=\ell_c^{-1}, \: c=1,2$. Histograms representing
       marginal posterior probability density of the learnt diagonal
       elements $\sigma_{11}^{(1)}$ and $\sigma_{22}^{(1)}$, of the
       covariance matrix $\bSigma_1$, are shown in the mid-row, left and
       middle panels (given respective data). Histograms representing marginals
       of the parameter
       $\rho=\displaystyle{\frac{\sigma_{12}}{\sqrt{\sigma_{11}^{(1)}\sigma_{22}^{(1)}}}}$ 
       are displayed in the mid-row right
       panel. Prediction of the values of the input parameter
       $\bS=(S_1,S_2)^T$ is possible only in the run performed with both
       training and test data. Marginals of $S_1$ and $S_2$ values learnt via
       MCMC-based sampling from the joint of all unknown parameters given
       $\bD^{\star}$, are shown in the lower panel, as approximated by histograms.}
\label{fig:nohypergp_with_wo_hist}
\end{figure}

\section{Results}
\label{sec:results}
\noindent
In this section, we present the results of learning the unknown
parameters of the 3rd-order tensor-normal likelihood, given the training as
well as the training+test data.

While Figure~1 of the Supplementary Materials and
Figure~\ref{fig:nohypergp_with_wo_hist} here depict results obtained from using the
$nonnested-GP$, in the following figures, results of the learning of all
relevant unknown parameters, using the $nested-GP$ model, are included. 
Figures that depict results from the $nested-GP$
approach will include results of the learning of amplitude $a_c$ and 
smoothing parameters $d_c:=1/\delta_c$ parameters. Also, our modelling
under the $nested-GP$ paradigm relies on a lookback-time $t_0$ which gives the 
number of iterations over which we gather the generated
$\ell_c$ values.

\begin{figure}
     \begin{center}
  {
       \includegraphics[width=10cm]{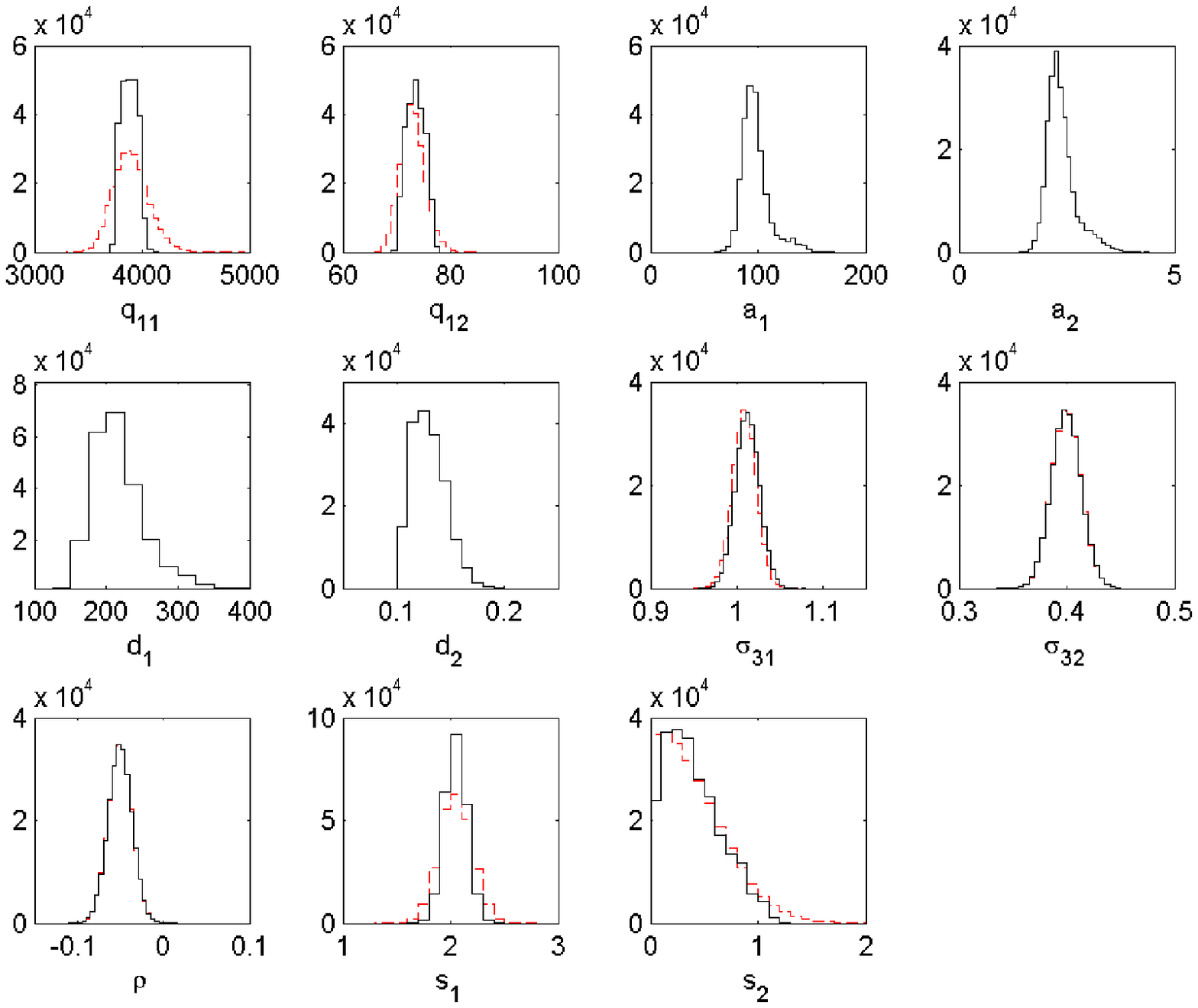}
  }
     \end{center}
     \caption{Results from run done with test+training data ${\bD}^{\star}$
       within the $nested-GP$ model, shown in black, as distinguished from the
       results of learning given the same data, and the $nonnested-GP$
       model depicted in grey (or red in the electronic copy of the
       thesis). Here the used value of $T_0$ is 200 iterations. Histograms
       approximating the marginal posterior probability densities of each
       sought unknown is depicted. Here, sought hyperparameter values $a_c$ and
       $\delta_c$ are relevant only to the $nested-GP$ model ($c=1,2$). Here, we
       have undertaken sampling from the joint posterior of all parameters,
       including the input parameter values $s_1^{(test)}$ and $s_2^{(test)}$,
       at which the test data are realised. Histograms approximating marginal
       posterior of each learnt unknown are presented.  }
\label{fig:hyper200_nohyper_gp}.
\end{figure}

\begin{figure}
     \begin{center}
  {
       \includegraphics[width=10cm]{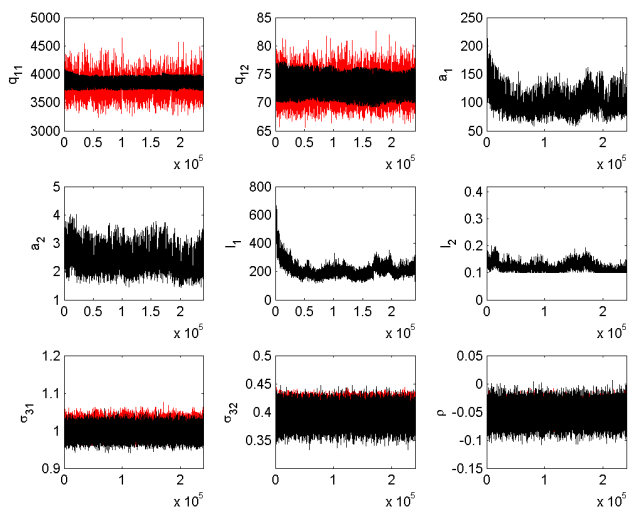}
  }
     \end{center}
     \caption{ Traces of parameters learnt using the training data ${\bf D}$,
       in the run performed with the $nonnested-GP$ model, are compared to
       traces of the corresponding parameter obtained in the run performed
       with the $nested-GP$ model. Traces of parameters learnt within the
       $nonnested-GP$ model are in grey (or red in the e-version) while the
       traces obtained using the $nested-GP$ model are shown in black.  }
\label{fig:hyper200_nohyper_gp_nopred_trace}.
\end{figure}

\subsection{Effect of discontinuity in the data, manifest in our results}
\noindent
One difference between the learning of parameters from the $nested-GP$, as
distinguished from the $nonnested-GP$ models is the quality of the inference,
in the sense that the uncertainty of parameters (i.e. the 95$\%$ HPDs) learnt
using the $nested-GP$ models, is less than that learnt using the $nonnested-GP$
models. This difference in the learnt HPDs is most marked for the learning of
values of $Q_1$ and $S_1$, and $S_2$ to a lesser extent. 

We explain this, by invoking the discontinuity in the training
data--distribution of $S_1$ in this data is sharply discontinuous, though
there is a less sharp discontinuity in the distribution of $S_2$ noted. We
refer to Figure~8 of \ctn{dc2007}, page 152. This figure is available at
{\url{https://www.aanda.org/articles/aa/pdf/2007/19/aa6677-06.pdf}}, and
corresponds to the base astronomical model used in the simulations that
generate the training data that we use here. This figure informs on the
distribution of location $\bS$; compatibility of the stellar velocity matrix
$\bv (=\bxi(\bs))$ realised (in astronomical simulations) at a given $\bs$, to
the test velocity matrix $\bv^{(test)}$ (recorded by the {\it Hipparcos}
satellite), is parametrised, and this compatibility parameter plotted against
$\bs$ in this figure. In fact, this figure is a contour plot of the distribution of such a
compatibility parameter, in the space ${\cal D}$, where $\bS\in{\cal
  D}\subset{\mathbb R}^2$. The 2 components of $\bS$ are represented in
polar coordinates, with $S_1$ the radial and $S_2$ the angular component. We
see clearly from this figure, that the distribution across $S_1$ is highly
discontinuous, at given values of $S_2$ (i.e. at fixed angular bins). In fact,
this distribution is visually more discontinuous, than the distribution across
$S_2$, at given values of $S_1$, i.e. at fixed radial bins (each of which is
represented by the space between two bounding radial arcs). In other words,
the velocity matrices that are astronomically simulated at different $\bS$
values, are differently compatible with a given reference velocity matrix
($\bv^{(test)}$)--and, the distribution of velocity matrix variable $\bV$, is
discontinuous across values of $\bS$, and in fact, less smoothly distributed
at fixed $s_2$, than at fixed $s_1$. Thus, this figure brings forth the
discontinuity with the input-space variable $\bS$, in the data tensor
$\bD_{\bV}$ that is part of the training data.

Then, it is incorrect to use a stationary kernel to parametrise the covariance
$\bSigma_3$, that informs on the covariance between velocity matrices
generated at different values of $\bS$. Our implementation of the $nested-GP$
model tackles this shortcoming of the model. However, when we implement the
$nonnested-GP$ model, Metropolis needs to explore a
wider volume of the state space to accommodate parameter values, given the
data at hand--and even then, there is a possibility for incorrect inference
under the stationary kernel model. This explains the noted trend of higher
95$\%$ HPDs on most parameters learnt using the $nonnested-GP$ model, compared
to the $nested-GP$ model, as observed in comparison of results from runs done
with training data alone, or both training and test data; compare
Figure~\ref{fig:hyper200_nohyper_gp} to
Figure~\ref{fig:hyper200_nohyper_gp_nopred_trace}, and note the comparison in
the traces as displayed in
Figure~\ref{fig:hyper200_nohyper_gp_nopred_trace}. Indeed, this also explains
the bigger difference noted in these figures when we compare the learning of
$q_1$ over $q_2$, in runs that use the stationary model, as distinguished from
the non-stationary model. After all, the discontinuity across $S_1$ is
discussed above, to be higher than across $S_2$.

\subsection{Effect of varying lookback times, i.e. length of historical data}
\noindent
To check for the effect of the lookback time $t_0$, we
present traces of the covariance parameters and kernel hyperparameters
learnt from runs undertaken within the $nested-GP$ model, but different $t_0$
values of 50 and 100, in Figure~\ref{fig:50_100_hypergp}, which we can
compare to the traces obtained in runs performed under the $nested-GP$ model,
with $t_0=200$, as displayed in
Figure~\ref{fig:hyper200_nohyper_gp_nopred_trace}.

\begin{figure}
     \begin{center}
  {
       \includegraphics[width=10cm]{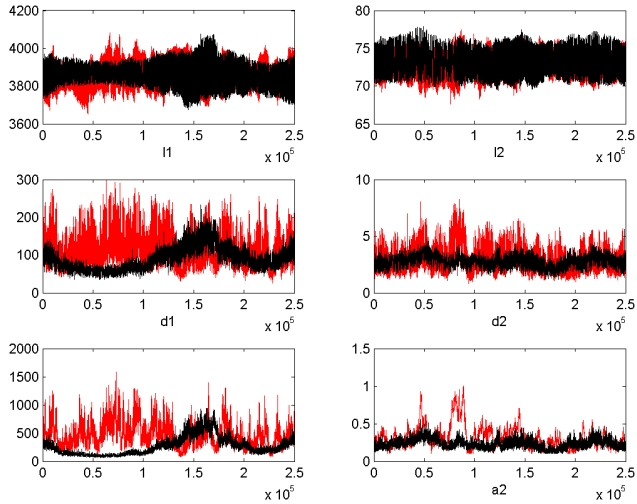}
  }
     \end{center}
     \caption{Comparison of traces of unknown smoothness parameters of
       $\bSigma_3$ and hyperparameters of GPs invoked to model these
       parameters, obtained in runs performed with training data ${\bf D}$ and
       $t_0=50$ (in grey, or red in the e-version) and $t_0=100$ (in black). 
       }
\label{fig:50_100_hypergp}.
\end{figure}

It is indeed interesting to note the trends in traces of the
the smoothness parameters $q$ that are 
the reciprocal of $\ell$ parameters, and values of the amplitude ($a_1,
a_2$) and values of length scale hyperparameters
($\delta_1, \delta_2$), evidenced in
Figure~\ref{fig:50_100_hypergp} and the results in black in
Figure~\ref{fig:hyper200_nohyper_gp_nopred_trace}). A zeroth-order model for 
these parameters that are realisations from a non-stationary process, is a moving averages time-series model--$MA(t_0)$ to be precise.
We note the increase in
fluctuation amplitude of the traces, with
decreasing $t_0$. For smaller values of lookback time $t_0$, the average
covariance between $g_{\bx_c}(t_1)$ and $g_{\bx_c}(t_2)$ is higher, than when $t_0$ is
higher, where the averaging is performed over a $t_0$-iteration long interval
that has its right edge on the current iteration; here
$\bx_c=(a_c,\delta_c)^T$, $c=1,2$ and as
introduced above, we model the length scale parameter of the kernel that
parametrises $\bSigma_3$, as $\ell_c=g_{\bx_c}(t)$. Here $g_{\bx_c}(\cdot)$ is modelled as
a realisation from a scalar-variate GP with covariance kernel that is
itself kernel-parametrised using an SQE kernel with amplitude $a_c$ and
correlation-length $\delta_c$. Then higher covariances between values of
$g_{\bx_c}(\cdot)$ at different $t$-values in general would suggest higher values of
the global amplitude of this parametrised kernel, and higher values of the
length-scales of this SQE kernel.

Indeed an important question is, what is the ``best'' $t_0$, given our
data. Such  question is itself of relevance, and discussed intensively under
distributed lag models, often within Econometrics \ctp{shirley}. An interesting trend noted in the parameter traces presented in
Figure~\ref{fig:50_100_hypergp} for $t_0=50,100$, and to a lesser extent for
$t_0=200$, in the results in black in
Figure~\ref{fig:hyper200_nohyper_gp_nopred_trace}, is the global near-periodic
existence of crests and troughs in these traces. This periodic fluctuation is
more marked for smoothness $q_1$ (=$1/\ell_1$) and the hyperparameters of the
scalar-variate GP used to model $g_{\bx_1}(\cdot)$, than for $q_2$ (and $a_2$
and $\delta_2$). 

From the point of view of a polynomial (of order $t_0$) model for the lag
operator--that transfers information from the
past $t_0$ realisations from a stochastic process to the current
iteration--the shape of the trace will be dictated by parameters of ths
model. If this polynomial admits complex roots, then coefficients of the relevant
lag terms will behave like a damped sine function with iterations. For a
different value of $t_0$, such a pronounced oscillatory trend might not be
equally apparent. Loosely
speaking, the value of $\ell_c$ in any iteration, represented by a moving
average, will manifest the result of superposition of the different
(discontinuous) modal neighbourhoods present in the data. The more multimodal
the data, i.e. larger the number of ``classes'' (by correlation-length scales)
of functional form $\bxi(\cdot)$ sampled from the tensor-variate GP, s.t.
superposition of the sample paths will cause a washing-out of the effect of
the different modes, and a less prominent global trend will be manifest in the
traces. However, for data that is globally bimodal, the superposition of the
two ``classes'' of sampled functions $\bxi(\cdot)$ will create a periodicity
in the global trend of the generated $\ell_c$ values (and thereby of the
smoothness parameter values $q_c$, where $q=\ell_c^{-1}$). 
Again, the larger the value $t_0$ of the lookback-time parameter, the moving
average is over a larger number of samples, and hence greater is the
washing-out effect. Thus, depending on the discontinuity in the data, it is
anticipated that there is a range of optimal lookback-time values, for which,
the global periodicity is most marked. This is what we might be noticing in
the trace of $q_1$ at $t_0=100$ displaying the global periodicity more
strongly than that at $t_0=200$ (see Figure~\ref{fig:50_100_hypergp} and
Figure~\ref{fig:hyper200_nohyper_gp_nopred_trace}). 

Another point is that the strength of this global periodicity will be stronger
for the correlation-length scale along that direction in input-space, the
discontinuity along which is stronger. Indeed, as we have discussed above, the
discontinuity in the data with varying $S_1$ is anticipated to be higher than
with $S_2$. So we would expect a more prominent periodic trend in the trace
of $q_1$ than $q_2$. This is indeed what to note in
Figure~\ref{fig:50_100_hypergp}. A simulation study can be undertaken to
explore the effects of empirical discontinuities.

The arguments above qualitatively explain the observed trends in the traces of
the hyperparameters, obtained from runs using different $t_0$. That in spite
of discrepancies in $a_c$ and $\delta_c$, with $t_0$, values of the length
scale parameter $\ell_c$ (and therefore its reciprocal $q_c$) are concurrent
within the 95$\%$ HPDs, is testament to the robustness of
inference. Stationarity of the traces betrays the achievement of convergence
of the chain.



We notice that the reciprocal correlation length scale $q_{1}$ is a
couple of orders of magnitude higher than $q_{2}$; correlation between
values of the sampled function $\bxi(\cdot)$, at 2 different $S_1$
values (at the same $s_2$), then wanes more quickly than correlation
between sampled functions computed at same $s_1$ and different $S_2$
values. Here $\bs=(s_1,s_2)^T$ and given that $\bS$ is the location of
the observer who observes the velocities of her neighbouring stars on
a two-dimensional polar grid, $S_1$ is interpreted as the radial
coordinate of the observer's location in the Galaxy and $S_2$ is the
observer's angular coordinate. Then it appears that the velocities
measured by observers at different radial coordinates, but at the same
angle, are correlated over shorter radial-length scales than
velocities measured by observers at the same radial coordinate, but
different angles. This is understood to be due to the astro-dynamical
influences of the Galactic features included by \ctn{dc2007}
in the simulation that generates the training data that we use
here. This simulation incorporates the joint dynamical effect of the
Galactic spiral arms and the elongated Galactic bar (made of stars)
that rotate at different frequencies (as per the astronomical model
responsible for the generation of our training data), pivoted at the
centre of the Galaxy. An effect of this joint handiwork of the bar and
the spiral arms is to generate distinctive stellar velocity
distributions at different radial (i.e. along the $S_1$ direction)
coordinates, at the same angle ($s_2$). On the other hand, the stellar
velocity distributions are more similar at different $S_2$ values, at
the same $s_1$. This pattern is borne by the work by \ctn{chakrabarty05}, in
which the radial and angular variation of the standard deviations of
these bivariate velocity distributions are plotted. Then it is
understandable why the correlation length scales are shorter along the
$S_1$ direction, than along the $S_2$ direction. 

Furthermore, for the correlation parameter $\rho$, physics suggests that the
correlation will be zero among the two components of a velocity vector. These
two components are after all, the components of the velocity vector in a
2-dimensional orthogonal basis. However, the MCMC chain shows that there is a
small (negative) correlation between the two components of the stellar
velocity vector.

\subsection{Predicting $\bs^{(test)}$}
\noindent
Figure~\ref{fig:nohypergp_with_wo_hist}, displays histogram-representations of
marginal posterior probability densities of the solar location coordinates
$s^{(test)}_1$, $s^{(test)}_2$; $q_{1}^{*}$ and $q_{2}^{*}$ that get updated
once the test data is added to augment the training data, and parameters
$\sigma_{11}^{1*}$, $\sigma_{22}^{1*}$ and $\rho^*$. 95$\%$ HPD credible
regions computed on each parameter in this inference scheme, are displayed in
Table~1 of Supplementary Materials. These figures display these parameters in the $nonnested-GP$ model. When the $nested-GP$ model is used, histogram-representations of the marginals of the aforementioned parameters, are displayed in Figure~\ref{fig:hyper200_nohyper_gp}.

Prediction of $\bs^{(test)}$ using the $nested-GP$ models gives rise to similar results as when the $nonnested-GP$ models are used, (see Figure~\ref{fig:hyper200_nohyper_gp} that compares the marginals of the solar location parameters sampled from the joint of all unknowns, given all data, in $nested-GP$ models, against those obtained when $nonnested-GP$ models are used).

The marginal distributions of $s_1^{(test)}$ indicates that the
marginal is unimodal and converges well, with modes at about 2 in model units.
The distribution of $s_2^{(test)}$ on the other hand is quite
strongly skewed towards values of $s_2^{(test)}\lesssim 1$ radians,
i.e. $s_2^{(test)}\lesssim 57$ degrees, though the probability mass in
this marginal density falls sharply after about 0.4 radians,
i.e. about 23 degrees. These values tally quite well with previous
work \ctp{chakrabarty2015bayesian}. In that earlier work, using the training data that we use in this work,
(constructed using the the astronomical model $sp3bar3{\_}18$
discussed by \ctn{chakrabarty2015bayesian}), the marginal distribution of
$s_1^{(test)}$ was learnt to be bimodal, with modes at about 1.85 and
2, in model units. The
distribution of $s_2^{(test)}$ found by \ctn{chakrabarty2015bayesian} is however more
constricted, with a sharp mode at about 0.32 radians (i.e. about 20
degrees). We do notice a mode at about this value in our inference,
but unlike in the results of \ctn{chakrabarty2015bayesian}, we do not find the
probability mass declining to low values beyond about 15 degrees. One
possible reason for this lack of compatibility could be that in
\ctn{chakrabarty2015bayesian}, the matrix of velocities $\bV$ was vectorised, so that
the training data then resembled a matrix, rather than a 3-tensor as
we know it to be. Such vectorisation could have led to some loss of correlation information, leading to their results.


Model checking of our models and results is undertaken in Section~3 of the
Supplementary Materials.

\subsection{Astronomical implications}
\label{sec:astro}
\noindent
The radial coordinate of the observer in the Milky Way, i.e. the solar radial
location, is dealt with in model units, but will need to be scaled to real
galactic unit of distance, which is kilo parsec (kpc). Now, from independent
astronomical work, the radial location of the Sun is set as 8 kpc. Then our
learnt value of $S_1^{(test)}$ is to be scaled to 8 kpc, which gives 1 model
unit of length to be ${{m}}:=\displaystyle{\left(\frac{8 \mbox{kpc}}{\mbox{learnt value of\:\:}S_1^{(test)}}\right)}$. Our main interest in learning the solar location is to find the frequency $\Omega_{bar}$ with which the Galactic bar is rotating, pivoted at the galactic centre, (loosely speaking). Here $\Omega_{bar}=\displaystyle{\frac{v_0}{\mbox{1 model unit of length}}=\frac{v_0}{{m}}}$, where $v_0=220$ km/s (see \ctn{dc2007} for details). The solar angular location being measured as the angular distance from the long-axis of the Galactic bar, our estimate of $S_2$ actually tells us the angular distance between the Sun-Galactic centre line and the long axis of the bar. These estimates are included in Table~\ref{tab:tab3}.

\begin{table*}[!h]
\caption{$95\%$ HPD on each Galactic feature parameter learnt from the solar
  location coordinates learnt using the two predictive inference schemes listed
above and as reported  in a past paper for the same training and test data.}
\label{tab:tab3}
\centering
\begin{tabular}{|l|l|l|}
\hline
          & $95\%$ HPD for $\Omega_{bar}$ (km/s/kpc)& for angular distance of\\
          &                                         &bar to Sun (degrees)\\ \hline
{\mbox{from posterior predictive}} & $[48.11,57.73]$        & $[4.53,43.62]$
   \\ \hline
{\mbox{from joint posterior}} & $[48.25,57.244]$           & $[2.25,46.80]$
        \\ \hline
{\mbox{from Chakrabarty et. al (2015)}} & $[46.75, 62.98]$ & $[17.60, 79.90]$
         \\ \hline
\end{tabular}
\end{table*}

Table~\ref{tab:tab3} displays the Galactic feature parameters that are derived
from the learnt solar location parameters, under the different inference
schemes using the $nonnested-GP$ model, namely, sampling from the joint posterior
probability of all parameters given all data, and from the posterior
predictive of the solar location coordinates given test data and GP parameters
already learnt from training data alone. The derived Galactic feature
parameters are the bar rotational frequency $\Omega_{bar}$ in the real
astronomical units of km/s/kpc and the angular distance between the bar and
the Sun, in degrees. The table also includes results from \ctn{chakrabarty2015bayesian}, the reference for which is in the main paper.

\section{Conclusions}
\noindent
Our work presents a method for learning tensor-valued functional relations
between a sytem parameter vector, and a tensor-valued observable, multiple
measurements of which build up a hypercuboidally-shaped data, that is in
general not continuous, thus demanding a non-stationary covariance structure
of the invoked GP. We clarify the need for generalising a stationary
covariance to one in which the hyperparameters (correlation length scales
along each direction of the space of the system parameter vector) need to be
treated as dependent on the sample function of the invoked GP. We address this
need by modelling the sought tensor-valued function with a tensor-variate GP,
each parameter of the covariance function of which, is modelled as a
dynamically varying, scalar-valued function that is treated as a realisation
from a scalar-variate GP with distinct covariance structure, that we
parametrise. We employ Metropolis-within-Gibbs-based inference, that allows
comprehensive and objective uncertainties on all learnt unknowns. Subsequent
to the learning of the sought tensor-valued function, we make an inverse
Bayesian prediction of the system parameter values at which test data on the
observable is realised. While in this work we focussed on the learning given
discontinuous data, the inclusion of non-stationarity in the covariance is a
generic cure for non-stationary data; we will consider an application to a
temporally varying, econometric dataset in a future contribution.

\renewcommand\baselinestretch{1.}
\small
\bibliographystyle{ECA_jasa}


\pagebreak

\begin{frontmatter}
\title{Supplementary Material for "Deep Bayesian Supervised Learning given Hypercuboidally-shaped,
Discontinuous Data, using Compound Tensor-Variate and Scalar-Variate Gaussian Processes"
}
\runtitle{Compound Tensor $\&$ Scalar-Variate GPs}

\begin{aug}
\author{
{\fnms{Kangrui} \snm{Wang}\thanksref{t2,m2}\ead[label=e2]{kw202@le.ac.uk}},
{\fnms{Dalia} \snm{Chakrabarty}\thanksref{t1,m1}\ead[label=e1]{d.chakrabarty@lboro.ac.uk}}
},
\thankstext{t2}{PhD student in Department of Mathematics, University of Leicester} 
\thankstext{t1}{Lecturer in Statistics, Department of Mathematical Sciences,
  Loughborough University}

\runauthor{Wang $\&$ Chakrabarty}

\affiliation{University of Leicester, Loughborough University}

\address{\thanksmark{m2} Department of Mathematics\\
University of Leicester\\
Leicester LE1 3RH,
U.K.\\
\printead*{e2}
}

\address{\thanksmark{m1} Department of Mathematical Sciences\\
Loughborough University\\
Loughborough LE11 3TU,
U.K.\\
\printead*{e1}
}

\end{aug}




\end{frontmatter}

Here, we refer to the main paper as ``KWDC''.

\section{Algorithm used to make inference in KWDC}
\begin{enumerate}
\item[1] In the $0$-th iteration, set all unknown parameters to arbitrarily
  chosen seed values : $a_c$ is set to the seed $a_{c}^{(0)}$ and $\delta_c$ is
  set to the seed $\delta_{c}^{(0)}$ $\forall c=1,\ldots,d$; $\sigma_q$ is set
  to the seed $\sigma_{q}^{(0)}$ $\forall q=1,\ldots, q_{max}$. We also set
  the length scales of the kernel-parametrised covariance matrix $\bSigma_p$
  to their respective seed values, i.e. set $\ell_c:=\ell_c^{(0)} \forall
  c=1,\ldots,d$.
\item[2(a)] At the beginning of the $t$-th iteration, for $t < t_0$, the
  current value of the $ell_c$ parameter is $\ell_c^{(t-1)}$. We propose the
  new value, $\ell_c^{(t\star)}$ from a Gaussian distribution, the mean of
  which is the current value of this parameter, namely $\ell_c^{(t-1)}$, and
  the variance of which is chosen experimentally, to be the constant $v_c$,
  i.e.
$$ \ell_c^{(t\star)} \sim {\cal N}(\ell_c^{(t-1)}, v_c).$$
This proposing is undertaken $\forall c=1,\ldots,d$. We choose adequate priors
(often Gaussian priors with mean $\ell_c^{(0)}$ and large constant variances)
on all the $\ell_c$ parameters. We refer to these priors as
$pi_0(\ell_1,\ldots,\ell_d)$. The proposed $\ell_c$ parameters then inform the kernel
function $K_p(\cdot,\cdot)$ that is used to kernel parametrise the covariance
matrix $\bSigma_p$, s.t. the proposed kernel-parametrised covariance matrix in
the $t$-th iteration, $t<t_0$, is $\bSigma_p^{(t\star)}=\bA_p^{(t\star)} (\bA_p^{(t\star)})^T$, while the current values of the $\ell_c$ parameters
suggest that the current kernel-parametrised covariance matrix is
$\bSigma_p^{(t-1)}=\bA_p^{(t-1)} (\bA_p^{(t-1)})^T$. We compute
$\bSigma_p^{(t\star)}=\left[\exp\left(-\frac{(\bs_i-\bs_j)^2}{2(\bell^{(t\star)})^2}\right)\right]$,
where
$\bell^{(t\star)}:=(\ell_1^{(t\star)},\ldots,\ell_d^{(t\star)})^T$. Similarly, 
$\bSigma_p^{(t-1)}$ is defined in terms of the vector of the length scales,
$\bell^{(t-1)}$ that is the current value of the end of the $t-1$-th iteration.
\item[3(a)] We compute the ratio of the posterior probability densities of the
  proposed $\ell_c$ parameters given the training data ${\bf D}$ to the
  posterior of the current $\ell_c$ values. The ratio of the proposal
  densities does not get invoked since the proposal is symmetric. Thus, the
  ratio that we compute is  
\small{
$$r :=\displaystyle{\frac{
 \exp(-\Vert ({\bf D})\times_1 \bA_1^{-1} \ldots \times_p
 (\bA_p^{(t\star)})^{-1} \ldots \times_k \bA_k^{-1}
 \Vert^2/2)\pi_0(\ell_1^{(t\star)},\ldots,\ell_d^{(t\star)})}
{
\exp(-\Vert ({\bf D})\times_1 \bA_1^{-1} \ldots \times_p
 (\bA_p^{(t-1)})^{-1} \ldots \times_k \bA_k^{-1}
 \Vert^2/2)\pi_0(\ell_1^{(t-1)},\ldots,\ell_d^{(t-1)})}},$$}
and compare $r$ with the value of the uniform random variate $U\sim{\cal
  U}[0,1]$.\\ 
--If $u \geq r$, we reject the proposed values
$\ell_1^{(t\star)},\ldots,\ell_d^{(t\star)}$, and set the current value of the
$\ell_c$ parameter at the end of the $t$-th iteration to be
$\ell_c^{(t)}=\ell_c^{(t-1)} \forall c=1,\ldots,d$.\\
--If $u < r$, we accept the proposed values
$\ell_1^{(t\star)},\ldots,\ell_d^{(t\star)}$, and set the current value of the
$\ell_c$ parameter at the end of the $t$-th iteration to be
$\ell_c^{(t)}=\ell_c^{(t\star)}, \forall c=1,\ldots,d$.\\
Thus, for iterations $t<t_0$, the first block update is a manifestation of
Random Walk.
\item[2(b)] If the iteration number $t$ is s.t. $t \geq t_0$, then we model
  the $\ell_c$ parameters, each as a realisation from a distinct
  scalar-variate GP, the covariance structure of which is kernel-parametrised
  s.t. these kernel hyperparameters are $a_c$ and $\delta_c$, $\forall
  c=1,\ldots,d$. Then the counterpart of point 2(a), within the ``nested GP''
  approach, is now discussed. Let the current values of $a_c$ and $\delta_c$
  be $a_c^{(t-1)}$ and $\delta_c^{(t-1)}$. We propose values of these
  parameters in the $t$-th iteration, respectively from a Truncated-Normal
  density (left-truncated at 0, mean $a_c^{(t-1)}$, and experimentally chosen
  constant variance $v_a^{(c)}$), and a Normal (mean $\delta_c^{(t-1)}$, and
  experimentally chosen constant variance $v_\delta^{(c)}$), i.e.
$$ a_c^{(t\star)} \sim {\cal TN}(a_c^{(t-1)}, 0, v_a^{(c)}),\quad\forall
c=1,\ldots,d, $$
$$ \delta_c^{(t\star)} \sim {\cal N}(\delta_c^{(t-1)}, 0,
v_\delta^{(c)}),\quad\forall c=1,\ldots,d. $$ Now the GP that $\ell_c$ is
modelled with, currently has a covariance structure that is parametrised by
the $t_0\times t_0$-dimensional covariance matrix $\bS_{c}$ s.t. currently the
$ij$-th element of this matrix is the covariance between the value of $\ell_c$
that was current in the $t-i$-th iteration and the value current in the
$t-j$-th iteration, i.e. $\bS_c^{(t-1)}=\left[ a_c^{(t-1)} \exp\left(-\frac{(i
      - j)^2}{2(\delta_c^{(t-1)})^2}\right)\right]$; $i,j=1,\ldots,t_0$. Thus,
at any fixed value (say $i$) of the input variable--the iteration number--the
$i$-th diagonal element $a_c^{(t-i)}$ of the covariance matrix $\bS$, gives
the variance of the Gaussian distribution that $\ell_c^{(t-i)}$ can be
considered to be sampled from. Following this, we reduce this scalar-variate
GP to a Gaussian distribution, by fixing the value of the input-space variable,
(which in this situation is the iteration number), to $t$. Then the proposed
variance of the Gaussian distribution that $\ell_c$ is sampled from, at the
$t$-th iteration, is the proposed value of the $a_c$ parameter in this
iteration, i.e. $a_c^{(t\star)}$. Under a Random Walk paradigm, the mean of
this Gaussian distribution is the current value of the $\ell_c$ parameter. In
other words, the model suggests that
$$\ell_c^{t\star}\sim{\cal N}(\ell_c^{(t-1)}, a_c^{(t\star)}).$$
This is essentially suggesting an adaptive Random Walk updating scheme for the
$\ell_c$ parameter, $\forall c=1,\ldots,d$.

\item[3(b)] This is the counterpart of point 3(a) for the $t\geq t_0$
  iterations, i.e. when the ``nested GP'' model is in play. Again, as during
  the discussion of 3(a), here we compute the ratio of the posterior
  probability densities of the proposed to the current values of the unknowns
  that are updated in the first block. This posterior has the contribution
  from the $k$-th ordered tensor-normal likelihood that the observable $\bV$
  ($=\bxi(\bS)$) is modelled as a realisation from. But the covariance matrix
  $\bSigma_p$ of this tensor-normal likelihood is kernel-parametrised, with
  a GP prior imposed on each length scale parameters $\ell_1,\ldots,\ell_d$ of
  this kernel function. Then the joint probability density of the set of
  $t_0$ number of realisations $\{\ell_c^{(t-t_0)},\ldots,\ell_c^{(t-1)}\}$ of
  the parameter $\ell_c$, from the 
  scalar-variate, zero-mean GP, is multivariate normal with mean vector ${\bf
    0}$ and covariance matrix $\bS_c=[s_c^{(ij)}]$, which is
  kernel-parametrised as $s_c^{ij} =
  \displaystyle{a_c\exp\left[-\frac{(i-j)^2}{2\delta_c^2}\right]}$, s.t. the
  current value of the covariance matrix in the $t$-th iteration is
  $\bS_c^{(t-1)}=\displaystyle{\left[{a_c^{(t-1)}\exp\left[-\frac{(i-j)^2}{2(\delta_c^{(t-1)})^2}\right]}\right]}$,
and the proposed value of the covariance matrix in the $t$-th iteration is 
$\bS_c^{(t\star)}=\displaystyle{\left[{a_c^{(t\star)}\exp\left[-\frac{(i-j)^2}{2(\delta_c^{(t\star)})^2}\right]}\right]}$. In other
  words, the prior probability density on the $t_0$-dimensional vector
  $\bell_c^{(t_0)}:=(\ell_c^{(t-t_0)},\ldots,\ell_c^{(t-1)})^T$ of
  values of the $c$-th length scale parameter, over the last $t_0$ iterations
  is multivariate normal, with mean vector ${\bf 0}$ and covariance matrix
  $\bS$, i.e.
$$\pi_0(\ell_c^{(t-t_0)},\ldots,\ell_c^{(t-1)}) =
\displaystyle{\frac{1}{\sqrt{\det(2\pi\bS)}}\exp\left[-\frac{1}{2}
    (\bell_c^{(t_0)})^T \bS^{-1}(\bell_c^{(t_0)})\right]},$$ where
$\bell_c^{(t_0)}$ and the current and proposed $\bS$ (as a function of current
and proposed $a_c$ and $\delta_c$ values) are defined
above. This is true $\forall c=1,\ldots,d$. Then the ratio of the posterior
probability density of the proposed to the current values of the unknowns
$a_1,\ldots,a_d,\delta_1,\ldots,\delta_c$, given the data is
{\tiny{
$$r :=\displaystyle{
\frac{
 \exp(-\Vert ({\bf D})\times_1 \bA_1^{-1} \ldots \times_p
 (\bA_p^{(t\star)})^{-1} \ldots \times_k \bA_k^{-1}
 \Vert^2/2)
\prod\limits_{c=1}^d
\frac{1}{\sqrt{\det(2\pi\bS^{(t\star)})}}
\exp\left[-\frac{1}{2}
          (\bell_c^{(t_0)})^T (\bS^{(t\star)}){-1}(\bell_c^{(t_0)})
    \right]
\prod\limits_{c=1}^d {\cal TN}(a_c^{(t\star)},0,v_a^{(c)})
}
{
\exp(-\Vert ({\bf D})\times_1 \bA_1^{-1} \ldots \times_p
 (\bA_p^{(t-1)})^{-1} \ldots \times_k \bA_k^{-1}
 \Vert^2/2) 
\prod\limits_{c=1}^d
\frac{1}{\sqrt{\det(2\pi\bS^{(t-1)})}}
\exp\left[-\frac{1}{2}
          (\bell_c^{(t_0)})^T (\bS^{(t-1)}){-1}(\bell_c^{(t_0)})
\right]
\prod\limits_{c=1}^d {\cal TN}(a_c^{(t-1)},0,v_a^{(c)})
}
},$$}}

and compare $r$ with the value of the uniform random variate 
$U\sim{\cal  U}[0,1]$.\\
--If $u \geq r$, we reject the proposed values of the unknowns, and set the current value of the
$\delta_c$ and $a_c$ parameters at the end of the $t$-th iteration to be
$a_c^{(t)}=a_c^{(t-1)}, \delta_c^{(t)}=\delta_c^{(t-1)},  \forall c=1,\ldots,d$.\\
--If $u < r$, we accept the proposed values,
and set $a_c^{(t)}=a_c^{(t\star)}, \delta_c^{(t)}=\delta_c^{(t\star)},  \forall c=1,\ldots,d$.\\ 
Thus, the updating for the $\delta_c$ parameters is Random Walk as they are
proposed from a Gaussian.
\item[4] In this point, we discuss the updating of the remaining unknowns, $\sigma_1,\ldots,\sigma_{q_{max}}$,
  i.e. the elements of covariance matrices of the tensor-normal
  joint probability distribution of a set of realisations of $\bV$
  ($=\bxi(\bS)$), that are not kernel-parametrised, but learnt directly by
  MCMC. These elements can in general be positive of negative, and so, in the
  $t$-th iteration, we
  propose them from a Gaussian with mean given by their current value
  $\sigma_q^{(t-1)}$, and experimentally fixed variance $v_q$,
  $q=1,\ldots,q_{max}$, i.e. the proposed value is
$$\sigma_q^{(t\star)}\sim {\cal N}(\sigma_q^{(t-1)}, v_q)\quad\forall
q=1,\ldots,v_q.$$
Then using these proposed values of the elements, the proposed values of all
covariance matrices other than $\bSigma_p$ that is kernel-parametrised, are
$\bSigma_1^{(t\star)},\ldots,\bSigma_{p-1}^{(t\star)},\bSigma_{p+1}^{(t\star)},\ldots,\bSigma_k^{(t\star)}$,
while their current values (populated by the current values $\sigma_q^{(t-1)}$
of elements) are
$\bSigma_1^{(t-1)},\ldots,\bSigma_{p-1}^{(t-1)},\bSigma_{p+1}^{(t-1)},\ldots,\bSigma_k^{(t-1)}$. The
priors on the $\sigma_q$ parameters are treated as Gaussians with mean given
by the seed value of $\sigma_q^{(0)}$ and experimentally chose, large
variance, to suggest vague priors. Thus, the ratio of the posterior
probability of the proposed and current $\sigma_1,\ldots,\sigma_{q_{max}}$,
parameters, given the training data ${\bf D}$, at the already updated
$\bSigma_p$ to value $\bSigma_p^{(t)}=(A_p^{(t)})^T A_p^{(t)}$, is
{\tiny{
$$r :=\displaystyle{
\frac{
    \exp(-\Vert ({\bf D})\times_1 (\bA_1^{(t\star)}){-1} \ldots \times_{p-1}
    (\bA_{p-1}^{(t\star)})^{-1}\times_{p} (\bA_{p}^{(t)})^{-1} \times_{p+1} (\bA_{p+1}^{(t\star)})^{-1} \ldots
    \times_k (\bA_k{(t\star)})^{-1}
    \Vert^2/2)
\pi_0(\sigma_1^{(t\star)},\ldots,\sigma_{q_{max}}^{(t\star)})
}
  { \exp(-\Vert ({\bf D})\times_1 (\bA_1^{(t-1)}){-1} \ldots \times_{p-1}
    (\bA_{p-1}^{(t-1)})^{-1}\times_{p} (\bA_{p}^{(t\star)})^{-1}\times_{p+1} (\bA_{p+1}^{(t-1)})^{-1} \ldots
    \times_k (\bA_k{(t-1)})^{-1}
    \Vert^2/2)
\pi_0(\sigma_1^{(t-1)},\ldots,\sigma_{q_{max}}^{(t-1)})
}
},$$}}
and compare $r$ with the value of the uniform random variate $U\sim{\cal
  U}[0,1]$.\\
It is possible that some of these $k-1$ covariance matrices are not learnt
using MCMC, but empirically estimated--in that case, the empirically estimated
value of the corresponding covariance matrix is used in both denominator and
numerator in the definition of the likelihood above, instead of its current
and proposed values respectively.
--If $u \geq r$, we reject the proposed values of the unknowns, and set the current value of the
$\sigma_q$ at the end of the $t$-th iteration to be
$\sigma_q^{(t)}=\sigma_q^{(t-1)}, \forall q=1,\ldots,q_{max}$.\\
--If $u < r$, we accept the proposed values,
and set $\sigma_q^{(t)}=\sigma_q^{(t\star)},  \forall q=1,\ldots,q_{max}$.\\ 
Thus, the updating for the $\sigma_q$ parameters is Random Walk as they are
proposed from a Gaussian.
\item[5] Repeat steps 2-5 until $t=t_{max}$, the length of the chain.
\end{enumerate}

\section{Results}
\label{sec:resutls}
\noindent
Figure~\ref{fig:nohypergp_with_wo} display traces of the sought parameters
learnt using the $nonnested-GP$. In the following figure
(Figure~\ref{fig:hyper200_nohyper_gp_nopred}), marginal posterior probability
density of the sought parameters, given training data, (depicted as
histograms), obtained using the $nonnoested-GP$ model, are compared to the
corresponding result obtained from the $nested-GP$ model. In this $nested-GP$
model, the covariance matrix $\bSigma_3$ (that bears information about the
covariance structure between sheets of data generated at different values of
the input variable $\bS=(S_1, S_2)^T$), is parameterised using a kernel, each
length-scale hyperparameter of which, is itself modelled as a
dynamically-varying function that is considered sampled from a GP.  For each
such scalar-variate GP that generates the length-scale $\ell_c$, $c=1,\ldots,
d=2$ the covariance matrix is itself kernel-parametrised using a stationary
kernel, with an amplitude parameter value $a_c$ and length-scale parameter
$\delta_c$.

\begin{figure}
     \begin{center}
  {
       \includegraphics[width=10cm]{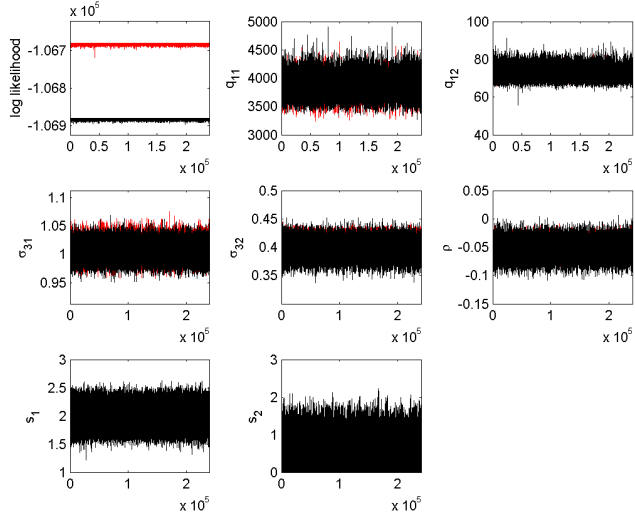}
  }
     \end{center}
     \caption{Results from run done with training data ${\bf D}$ with the $nonnested-GP$ model, are shown in grey (or red in the
       electronic copy of the paper) while results from run undertaken with
       training and test data, $\bD^{\star}$, in this $nonnested-GP$ model,
       are depicted in black. Traces of the logarithm of the likelihood are
       displayed from the two runs in the top left panel. Reciprocal of the
       length scale parameters are the shown in the top middle and right
       panels; here $q_c=\ell_c^{-1}, \: c=1,2$. Traces of the learnt diagonal
       elements $\sigma_{11}^{(1)}$ and $\sigma_{22}^{(1)}$, of the
       covariance matrix $\bSigma_1$, are shown in the mid-row, left and
       middle panels. Trace of the correlation
       $\rho=\displaystyle{\frac{\sigma_{12}}{\sqrt{\sigma_{11}^{(1)}\sigma_{22}^{(1)}}}}$ 
       is displayed in the mid-row right
       panel. Prediction of the values of the input parameter
       $\bS=(S_1,S_2)^T$ is possible only in the run performed with both
       training and test data. Traces of $S_1$ and $S_2$ values learnt via
       MCMC-based sampling from the joint of all unknown parameters given
       $\bD^{\star}$, are shown in the lower panel.  
       }
\label{fig:nohypergp_with_wo}.
\end{figure}

\begin{figure}
     \begin{center}
  {
       \includegraphics[width=10cm]{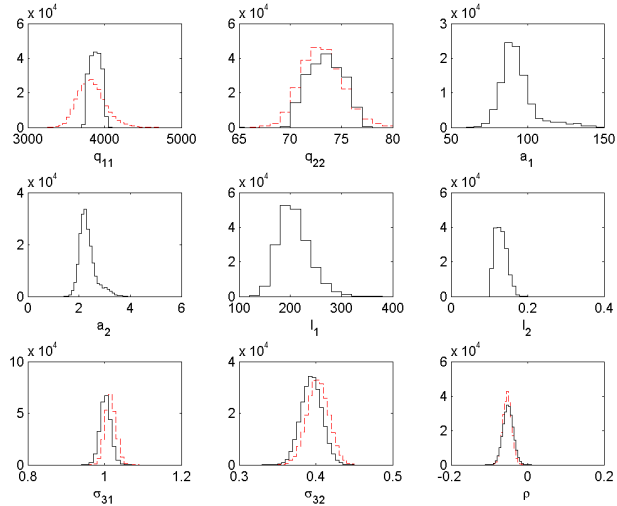}
  }
     \end{center}
     \caption{Marginal posterior probability densities of unknown parameters,
       given training data ${\bf D}$, are depicted as histograms. Histograms
       obtained from the run done with only 
       within the $nested-GP$ model, shown in black, as distinguished from the
       results of learning given the same data, and the $nonnested-GP$
       model depicted in grey (or red in the
       electronic copy of the thesis). Given the data used here, $s_1^{(test)}$ and
$s_2^{(test)}$, are not learnt.
       }
\label{fig:hyper200_nohyper_gp_nopred}.
\end{figure}

95$\%$ HPD credible regions computed on each learnt parameter given
the $nonnested-GP$ model, are displayed in Table~\ref{tab:tab1}. 
Again, a
similar set of results from the chains run with the $nested-GP$ models 
are displayed in Table~\ref{tab:tab2}. The results
on prediction of $\bs^{(test)}$ are
also presented in Table~\ref{tab:tab1} and Table~\ref{tab:tab2}.

\begin{table}[htbp]
\caption{$95\%$ HPD credible regions on each learnt parameter, from the
  $nonnested-GP$ model}
\label{tab:tab1}
\centering
\begin{tabular}{|l|l|l|l|}
\hline
{\mbox{Parameters}} & {\mbox{using only training data}}      & {\mbox{sampling from posterior predictive}} & {\mbox{sampling from joint}}         \\ \hline
$q_{1}$      & {[}3492.1,4198.1{]}   &                & {[}3573.2,4220.8{]}   \\ \hline
$q_{2}$      & {[}68.92,76.88{]}     &                & {[}68.37,77.33{]}     \\ \hline
$\sigma^{(1)}_{11}$ & {[}0.9837,1.0380{]}   &                & {[}0.9797,1.0338{]}   \\ \hline
$\rho$              & {[}-0.0653,-0.0275{]} &                & {[}-0.0798,-0.0261{]} \\ \hline
$\sigma^{(1)}_{22}$ & {[}0.3747,0.4234{]}   &                & {[}0.3703,0.4237{]}   \\ \hline
$s_1$               & -                     & {[}1.8212,2.1532{]} & {[}1.8038,2.1960{]} \\ \hline
$s_2$               & -                     & {[}0.0421,1.2052{]}  & {[}0.0157,1.2172{]}   \\ \hline
\end{tabular}

\end{table} 

\begin{table}[htbp]
\caption{$95\%$ HPD credible regions on each learnt parameter, from the
  $nested-GP$ model}
\label{tab:tab2}
\centering
\begin{tabular}{|l|l|l|l|}
\hline
{\mbox{Parameters}} & {$t_0=200$}      & {$t_0=100$} & {$t_0=50$}         \\ \hline
$q_{1}$      & [3740.96, 3917.32] & [3710.4, 4011.66] & [3650.92, 4033.51]  \\ \hline
$q_{2}$      & [70.34, 75.70] & [70.42, 76.43] & [68.94, 76.22]    \\ \hline
$a_{1}$      & [78.67, 124.02] & [43.82, 167.35] & [48.27, 219.37]  \\ \hline
$a_2$        & [1.88, 3.03] & [2.12, 3.57] & [1.64, 6.16] \\ \hline
$d_1$ & [155.64, 301.65] & [78.47, 521.67] & [123.42, 828.37] \\ \hline
$d_2$ & [0.10, 0.15] & [0.12, 0.46] & [0.10, 0.52]\\ \hline
$\sigma_{31}$ & [0.97, 1.02] & [0.97, 1.03] & [0.98, 1.02]\\ \hline
$\sigma_{32}$ & [0.37, 0.41] & [0.37, 0.41] & [0.38, 0.41]\\ \hline
$\rho$ &  [-0.076, -0.031] & [-0.073, -0.03] & [-0.075, -0.032]\\ \hline
$s_1$ & [1.83, 2.16] &[1.77, 2.22] &[1.76, 2.24] \\ \hline
$s_2$ & [0.138, 1.15] &[0.112, 1.16] & [0.071, 1.15]\\ \hline
\end{tabular}

\end{table}

\section{Model Checking}
\label{sec:model_ch}
\noindent
One way to check for the model and results, given the data at hand, is to
generate data from the learnt model, and then compare this generated data with
the observed data. Now, the model that we learn, is essentially the
tensor-variate GP that is used to model the functional relationship
$\bxi(\cdot)$ between the observable $\bV$ and the input-space parameter
$\bS$. By, saying that we want to generate new data, we imply the
prediction of a new value of $\bV$, given the learnt model of this GP. 

This prediction of new datum on $\bV$, is fundamentally different from the
inverse prediction of the value $\bs^{(test)}$ of the input-space parameter
$\bS$ that we have undertaken--as discussed above--where the sought
$\bs^{(test)}$ is the value of $\bS$ at which test data $\bv^{(test)}$ on
$\bV$ is recorded. There is no closed-form solution to the posterior
predictive of $\bs^{(test)}$ given the test data and the learnt GP parameters.

In fact, at chosen values of $\bS$--chosen to be the design points in the
training data, for convenience--the covariance function $\bSigma_3$ of this
GP, (modelled as a GP with an estimated mean), is known, given the learnt
values of the parameters of the kernel used to parametrise
$\bSigma_3$. However, in our Bayesian inference, we do not really learn a value
of any parameter, but learn the marginal posterior of each unknown parameter,
given the data. Thus, in order to pin the value of each element of $\bSigma_3$, we
identify the parameter value corresponding to a selected summary of this
posterior distribution. For example, we could choose to define $\bSigma_3$ at
pairs of known design points $\bs_i$, $\bs_j$, and the modal value of
$\ell_c$--identified from the marginal posterior of $\ell_c$ inferred upon,
given the data. Here $i,j\in\{1,\ldots,n=216\}$. The resulting value of the
$ij$-th element of $\bSigma_3$ will then provide one summary, of the
covariance between the $50\times 2$ stellar velocity matrix $\bv_i$ realised
at $\bS=\bs_i$, and $\bv_j$ realised at $\bS=\bs_j$. Similarly, the learnt
modal values of the parameters $\sigma_{11}^{(1)}$, $\sigma_{22}^{(1)}$ and
$\rho$ define one summary of the covariance matrix $\bSigma_1$ that informs on
the covariance between the 2 $216\times 50$-dimensional sheets of data on each
component of the 2-dimensional stellar velocity vector. Again, other summaries
of the parameter values could be used as well, for example, the parameter
value identified at the mean of the marginal posterior density of this
parameter, as learnt given the training data, is also used.

In this model checking exercise, the unknowns are certain elements of the
cuboidally-shaped data comprising the 216 number of $50\times 2$-dimensional
stellar velocity matrices generated by astronomical simulation, at chosen
design points $\bs_1,\ldots,\bs_{216}$, i.e. the 3rd-order tensor ${\bf
  D}_V:=\{ \bv_1, \vdots \bv_2, \vdots \ldots, \vdots \bv_{216} \}$. In the
first attempt to model checking, we generate all elements of the $q$-th such
simulated stellar velocity matrix $\bv_q$, (that is generated at the known
design point $\bs_q$), i.e. generate values of $50\times 2=100$ unknown
elements of matrix $\bv_q$. We refer to these
unknown elements of $\bv_q$ as $v_{11}^{(q)}, v_{12}^{(q)},
v_{21}^{(q)},\ldots, v_{50,2}^{(q)}$. The 3rd-ordered tensor without the $q$-th
slice, is referred to as ${\bf
  D}_V^{(-q)}:=\{ \bv_1, \vdots \bv_2, \vdots \ldots,
\vdots\bv_{q-1},\vdots\bv_{q+1}, \vdots \bv_{216} \}$. The joint posterior
probability density of the 100 unknowns, at the learnt modal values
$q_1^{(mode)},q_2^{(mode)},\sigma_{11}^{(1,mode)},\sigma_{22}^{(1,mode)},\rho^{(mode)}$
is
$$\pi\left(v_{11}^{(q)}, v_{12}^{(q)}, v_{21}^{(q)},\ldots,
  v_{50,2}^{(q)}\vert {\bf D}_V^{(-q)}\right) \propto {\cal TN}_{2\times
  50\times 216}({\hat{\bM}}, \bSigma_1^{(mode)}, {\hat{\bSigma}}_2, \bSigma_3^{(mode)}),$$
where, \\
--the 3rd-ordered tensor-valued data that enters the parametric form of the
3rd-ordered tensor-normal density on the RHS, has elements of its $q$-th
slice, (out of a total of 216 slices), unknown. All other elements of this
$2\times 50\times 216$-dimensional tensor are known; 
--uniform priors are used on the unknowns; 
--$\bSigma_1^{(mode)}$ is the
learnt modal value of the $2\times 2$-dimensional covariance matrix
$\bSigma_1$ s.t. its $1,1$-th element is $\sigma_{11}^{(1,mode)}$, $2,2$-th
element is $\sigma_{22}^{(1,mode)}$, $1,2$-th element is
$\rho^{(mode)}\sqrt{\sigma_{22}^{(1,mode)}\sigma_{11}^{(1,mode)}}$, and the
$2,1$-th element is equal to the $1,2$-th element (as this is a covariance matrix);\\  
--$\bSigma_3^{(mode)}$ is the learnt modal value of the $216\times
216$-dimensional covariance matrix $\bSigma_3$, s.t. its $ij$-th element is
$\displaystyle{\exp\left[-(\bs_i-\bs_j)^T \bQ^{(mode)} (\bs_i-\bs_j)\right]}$,
with the non-zero elements of the diagonal $2\times 2$-dimensional
$\bQ^{(mode)}$-matrix given by $q_1^{(mode)}$ and $q_2^{(mode)}$. $\bs_i$
being the $i$-th design point, is known $\forall i,j=1,\ldots,216$.

\begin{figure}
\vspace*{-.9in}
     \begin{center}
$\begin{array}{c c c}
  {
       \includegraphics[width=5cm]{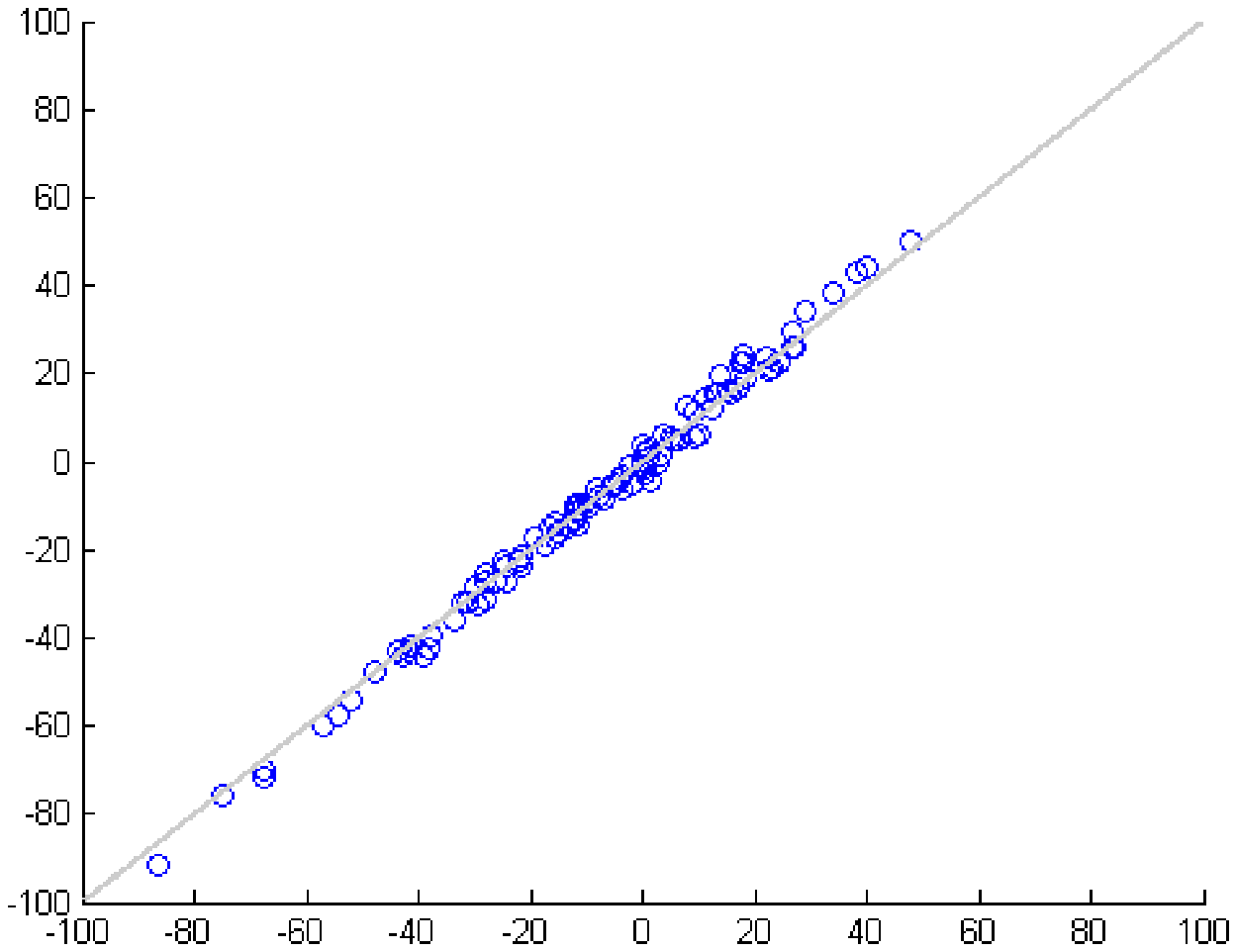}
       \includegraphics[width=5cm]{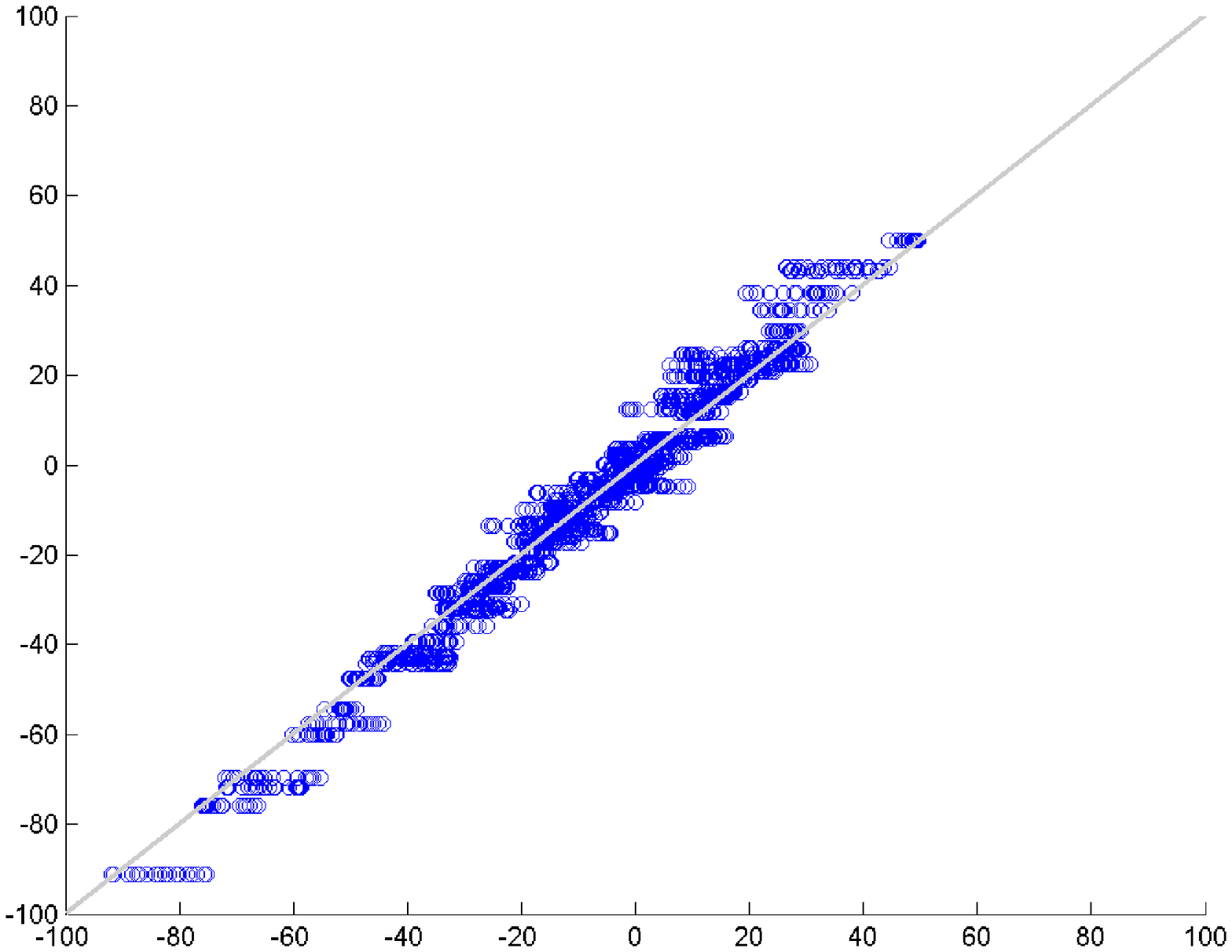}
       \includegraphics[width=5cm]{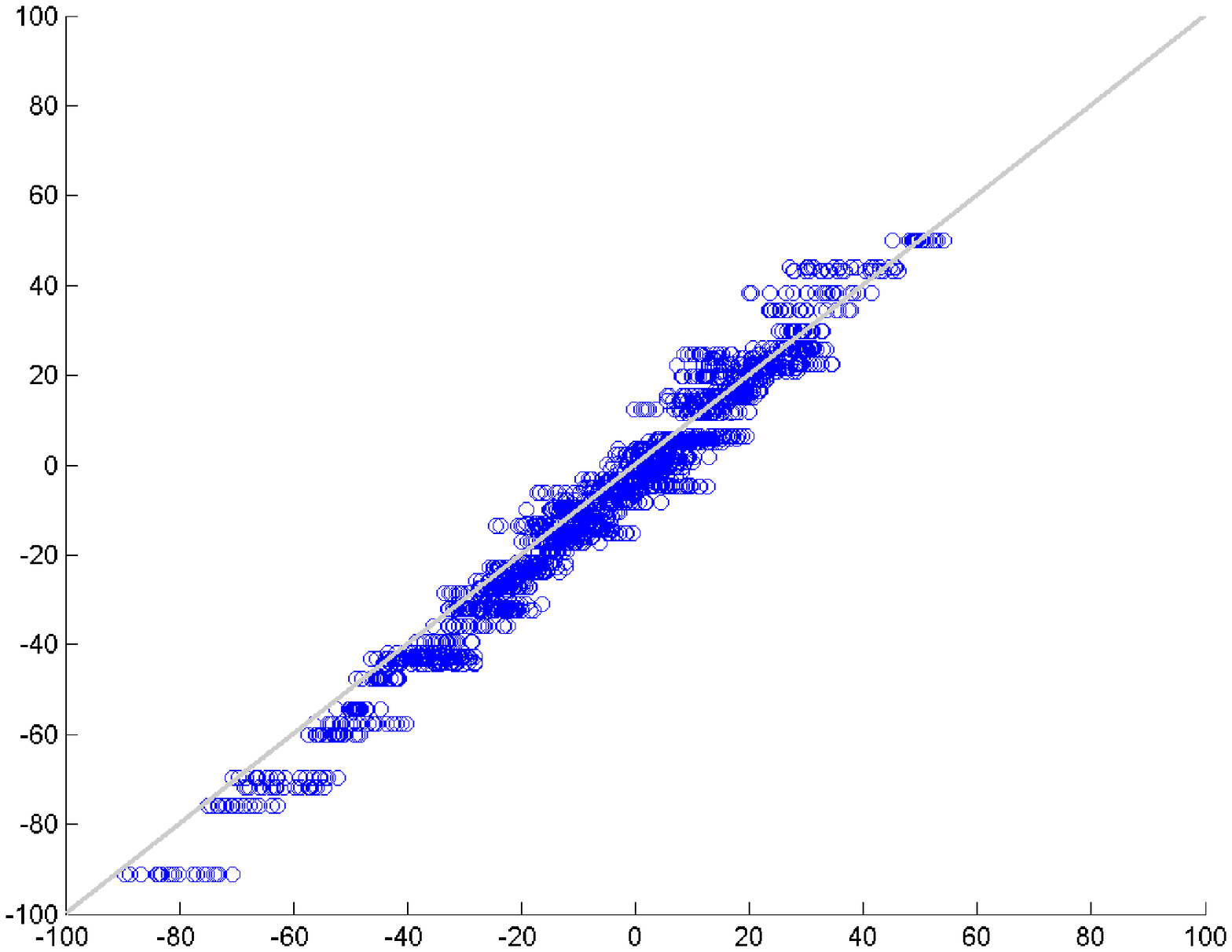}
  }
\end{array}$
     \end{center}
     \caption{{\it Left:} Comparison of the observed and predicted values of
       elements of the $q$-th $50\times 2$-dimensional stellar velocity matrix
       $\bv_q$, where 216 such matrices constitute the training data ${\bf
         D}_V$ (on velocities of 50 stellar neighbours of the Sun) that is
       generated by astronomical simulations. The predicted or learnt values
       are obtained from a RW-MCMC chain undertaken with the all elements of
       the 3rd-order tensor ${\bf D}_V$ known, except for the elements of its
       $q$-th slice, and the learnt values of the parameters of the GP used to
       model the data at hand, at a chosen summary, namely the mode, of the
       marginal posterior density of each such learnt GP parameter. Here
       $q$=200. Equality of the observed and predicted values of he elements
       of $\bv_q$ is indicated by the point lying on the drawn straight line
       with unit slope; the predicted values are found to lie close to this
       line. {\it Middle:} Depicts a similar comparison, as displayed in the
       left panel, but for 20 distinct values of $q$, namely for
       $q=190,191,\ldots,210$. {\it Right:} Depicts the same comparison of
       observed and predicted values of elements of 20 slices
       $\bv_{190},\ldots,\bv_{210}$, but this time, the employed GP parameters
       are the means of their respective marginals. Thus, this model-checking
       exercise checks for the used models and results obtained (given the
       data at hand) at the mean of the respective posterior.}
\label{fig:mod_chk}.
\end{figure}

To learn the 100 unknowns $v_{11}^{(q)}, v_{12}^{(q)}, v_{21}^{(q)},\ldots,
v_{50,2}^{(q)}$, we run a RW Metropolis-Hastings chain, with the data defined
as above, the known 216 number of design points, and all the learnt, modal
parameter values. The joint posterior of the unknowns that defines the
acceptance ratio in this chain, is given as in the last equation. The chain is
run for 20,000 iterations, for $q$=200, and the mean of the last 1000
samples of $v_{ij}^{(200)}$ is recorded, where $i=1,\ldots,50$, $j=1,2$. 
These sample means ${\bar{v}}_{ij}^{(200)}$ then constitute the learnt value
of the 100 elements of the 200-th stellar velocity matrix $\bv_{200}$. We plot
the pairs of learnt value ${\bar{v}}_{ij}^{(200)}$ of elements  of the
$\bv_{200}$ matrix, against the empirically observed value of this element,
$\forall i=1,\ldots,50$, $\forall j=1,2$. The plot is presented in the left
panel of Figure~\ref{fig:mod_chk}. Thus, each point on this plot is a pair
$({\mbox{empirically observed value of}}\:{{v}}_{ij}^{(200)},
{\bar{v}}_{ij}^{(200)})$, and there are $50\times 2=100$ points in this plot. 
The points are found to lie around the straight line with slope 1. In other
words, the values of the elements in the $q$-th (=200-th) slice of the
training data that we learn using our model, are approximately equal
to the empirically observed values of these elements. This is corroboration of
our models and results.

We attempt a similar prediction of elements of the training data for other
values of $q$, namely for $q=190,\ldots,210$. The learnt values of elements
of $\bv_q$, for each $q$, is plotted against the empirically observed elements
of $\bv_q$. We have superimposed results for all 20 values of $q$ in the same
plot, resulting in the middle panel of Figure~\ref{fig:mod_chk}. Again, the
values predicted for all 20 slices, are found to be close to the empirical
observations, as betrayed by the points lying close to the straight line of
unit slope.

Lastly, we wanted to ensure that the encouraging results from our model
checking exercise is robust to changes in the posterior summary of the learnt
GP parameters. Thus, we switch to using the mean of the parameter marginal
posterior from the posterior mode, and carry out the same exercise of
predicting elements of slices $\bv_{190},\ldots,\bv_{210}$. Results are
displayed in the right panel of Figure~\ref{fig:mod_chk}. Again, very
encouraging corroboration of our used models and results (of learning the GP
parameters) is noted. Indeed, in such model checking exercises, encouraging
match between the predictions and the empirical observations lends confidence
in the used models and results obtained therefrom, given the data at
hand--such models and results are the inputs to this exercise. However, if
lack of compatibility is noted in such a model checking exercise, between  
empirical observations and predictions, then it implies that either the used
modelling is wrong, and/or the results obtained therefrom given the data are
wrong. However, the model checking exercise that we undertake, vindicates our
models and results, given the data at hand.

\section{Proof of Theorem~3.3 in KWDC}
\begin{theorem}
\label{th:main1}
{
Given $\ell_c=g_{c,\bx}(t)$, with $T\in{\cal N}\subset{\mathbb Z}_{\geq 0}$ and
$\ell_c\in{\mathbb R}$, 
the map $g_{c,\bx}:{\cal N}\longrightarrow{\mathbb R}_{\geq 0}$ is a Lipschitz-continuous map,
$\forall c=1,\ldots,d$, where we know that there is a distinct value of
$\ell_c$ generated at a given $t$, i.e. $g_{c,\bx}(\cdot)$ is injective. Here
${\cal N}=\{t-1,\ldots,t-t_0\}$.}
\end{theorem}
\begin{proof}
  { Distance $d_{t_{1,2}}$ between $t_1,t_2\in{\cal N}$ is $\vert
    t_1-t_2\vert$.

    Similarly, distance $d_{g_{1,2}}$ between
    $g_{c,\bx}(t_1),g_{c,\bx}(t_2)\in{\mathbb R}_{\geq 0}$ is $\vert
    g_{c,\bx}(t_1)-g_{c,\bx}(t_2)\vert$.

    $${\mbox{Assume}}\quad \displaystyle{\frac{d_{g_{1,2}}}{M}} > d_{t_{1,2}},\quad
    M>0,\:\:M\:\:{\mbox{is finite}}\:\:\forall
    t_1,t_2\in{\cal N}.$$

    \noindent
    Let $\displaystyle{\frac{d_{g_{1,2}}}{M}}:= 1/2$. \\
    $\Longrightarrow \vert t_1-t_2\vert < 1/2$ by  our assumption, \\
    i.e. for this choice of the LHS of the inequation assumed, the only
    solution for $\vert t_1-t_2\vert < 1/2$ is $t_1=t_2$. \\
    But $t_1=t_2\Longrightarrow g_{c,\bx}(t_1)=g_{c,\bx}(t_2)$ for
    injective $g_{c,\bx}(\cdot)$, 
    i.e. LHS of the assumed inequation is then 0.\\ 
    This is a contradiction (contradicts our choice of 1/2 for the LHS).\\
    $\therefore$ our assumption is wrong,\\
    $\Longrightarrow$, the correct inequation is: 
$$\displaystyle{\frac{d_{g_{1,2}}}{M}} \leq d_{t_{1,2}},\quad
    M>0,\:\:M\:\:{\mbox{is finite}}\:\:\forall
    t_1,t_2\in{\cal N},$$
i.e.
$$\displaystyle{{\vert g_{c,\bx}(t_1)-g_{c,\bx}(t_2)\vert}} \leq M \vert t_1-t_2\vert,\quad
    M>0,\:\:M\:\:{\mbox{is finite}}\:\:\forall
    t_1,t_2\in{\cal N}\subset{\mathbb Z}_{\geq 0}$$
i.e. $g_{c,\bx}(\cdot)$ is Lipschitz continuous.
     }
\end{proof}

\renewcommand\baselinestretch{1.}
\small
\end{document}